\documentclass[prx,superscriptaddress,aps, twocolumn, 10pt,longbibliography, floatfix]{revtex4-2}

\usepackage{tikzit}

\tikzstyle{colored_text_box}=[
    shape=rectangle,
    rounded corners=3pt,
    fill={rgb,255: red,255; green,255; blue,255},  
    fill opacity=0.6,
    text opacity=1,
    draw=none,
    inner sep=2pt,
    font={\fontsize{10pt}{12pt}\selectfont},
    text=orange     
]

\tikzstyle{pink_text_box}=[
    shape=rectangle,
    rounded corners=3pt,
    fill={rgb,255: red,255; green,255; blue,255},  
    fill opacity=0.6,
    text opacity=1,
    draw=none,
    inner sep=2pt,
    font={\fontsize{10pt}{12pt}\selectfont},
    text={rgb,255: red,255; green,20; blue,147}     
]
\tikzstyle{Label}=[font={\fontsize{10pt}{12pt}\selectfont}, tikzit category=circuit, align=center]

\tikzstyle{gateTest}=[shape=rectangle, fill=white, draw=black, minimum height=10mm, minimum width=7.5mm, font={\fontsize{10pt}{12pt}\selectfont}, tikzit category=circuit, inner sep = 0]

\tikzstyle{gateTest2}=[shape=rectangle, fill=white, draw=black, minimum height=17.5mm, minimum width=7.5mm, font={\fontsize{10pt}{12pt}\selectfont}, tikzit category=circuit, inner sep = 0]

\tikzstyle{gate}=[shape=rectangle, text height=1.5ex, text depth=0.25ex, yshift=0.5mm, fill=white, draw=black, minimum height=5mm, yshift=-0.5mm, minimum width=5mm, font={\small}, tikzit category=circuit]
\tikzstyle{gate_huge}=[shape=rectangle, text height=1.5ex, text depth=0.25ex, yshift=0.5mm, fill=white, draw=black, minimum height=10mm, yshift=-0.5mm, minimum width=50mm, font={\small}, tikzit category=circuit]
\tikzstyle{gate_wide}=[shape=rectangle, text height=1.5ex, text depth=0.25ex, yshift=0.5mm, fill=white, draw=black, minimum height=5mm, yshift=-0.5mm, minimum width=15mm, font={\small}, tikzit category=circuit]
\tikzstyle{gate_mwide}=[shape=rectangle, text height=1.5ex, text depth=0.25ex, yshift=0.5mm, fill=white, draw=black, minimum height=5mm, yshift=-0.5mm, minimum width=8mm, font={\small}, tikzit category=circuit]
\tikzstyle{gate_vwide}=[shape=rectangle, text height=1.5ex, text depth=0.25ex, yshift=0.5mm, fill=white, draw=black, minimum height=5mm, yshift=-0.5mm, minimum width=20mm, font={\small}, tikzit category=circuit]
\tikzstyle{gate_high}=[shape=rectangle, text height=1.5ex, text depth=0.25ex, yshift=0.5mm, fill=white, draw=black, minimum height=15mm, yshift=-0.5mm, minimum width=5mm, font={\small}, tikzit category=circuit]
\tikzstyle{gate_vhigh}=[shape=rectangle, text height=1.5ex, text depth=0.25ex, yshift=0.5mm, fill=white, draw=black, minimum height=25mm, yshift=-0.5mm, minimum width=5mm, font={\small}, tikzit category=circuit]
\tikzstyle{big gate}=[shape=rectangle, text height=1.5ex, text depth=0.25ex, yshift=0.5mm, fill=white, draw=black, minimum height=10mm, yshift=-0.5mm, minimum width=5mm, font={\small}, tikzit category=circuit]
\tikzstyle{Z dot}=[inner sep=0mm, minimum size=2mm, shape=circle, draw=black, fill={rgb,255: red,221; green,255; blue,221}, tikzit category=zx]
\tikzstyle{Z phase dot}=[minimum size=5mm, font={\footnotesize\boldmath}, shape=rectangle, rounded corners=2mm, inner sep=0.2mm, outer sep=-2mm, scale=0.8, tikzit shape=circle, draw=black, fill={rgb,255: red,221; green,255; blue,221}, tikzit draw=blue, tikzit category=zx]
\tikzstyle{X dot}=[Z dot, shape=circle, draw=black, fill={rgb,255: red,255; green,136; blue,136}, tikzit category=zx]
\tikzstyle{X phase dot}=[Z phase dot, tikzit shape=circle, tikzit draw=blue, fill={rgb,255: red,255; green,136; blue,136}, font={\footnotesize\boldmath}, tikzit category=zx]
\tikzstyle{hadamard}=[fill=yellow, draw=black, shape=rectangle, inner sep=0.6mm, minimum height=1.5mm, minimum width=1.5mm, tikzit category=zx]
\tikzstyle{paulibox}=[fill={rgb,255: red,221; green,221; blue,255}, draw=black, shape=rectangle, inner sep=0.6mm, minimum height=5mm, minimum width=5mm, font={\footnotesize}, text height=1.5ex, text depth=0.25ex, tikzit category=zx]
\tikzstyle{vertex}=[inner sep=0mm, minimum size=1mm, shape=circle, draw=black, fill=black, tikzit category=misc]
\tikzstyle{vertex set}=[inner sep=0mm, minimum size=1mm, shape=circle, draw=black, fill=white, font={\footnotesize\boldmath}, tikzit category=misc]
\tikzstyle{small black dot}=[fill=black, draw=black, shape=circle, inner sep=0pt, minimum width=1.2mm, tikzit category=circuit]
\tikzstyle{cnot ctrl}=[fill=black, draw=black, shape=circle, inner sep=0pt, minimum width=1.2mm, tikzit category=circuit]
\tikzstyle{cnot targ}=[fill=white, draw=white, shape=circle, tikzit category=circuit, label={center:$\oplus$}, inner sep=0pt, minimum width=2.1mm, tikzit fill={rgb,255: red,102; green,204; blue,255}, tikzit draw=black]
\tikzstyle{ket}=[fill=white, draw=black, shape=regular polygon, regular polygon sides=3, regular polygon rotate=-30, scale=0.7, inner sep=1pt, tikzit category=circuit, tikzit shape=rectangle, tikzit fill=green]
\tikzstyle{bra}=[fill=white, draw=black, shape=regular polygon, regular polygon sides=3, regular polygon rotate=30, scale=0.7, inner sep=1pt, tikzit category=circuit, tikzit shape=rectangle, tikzit fill=red]
\tikzstyle{scalar}=[shape=rectangle, text height=1.5ex, text depth=0.25ex, yshift=0.5mm, fill=white, draw=black, minimum height=5mm, yshift=-0.5mm, minimum width=5mm, font={\small}]
\tikzstyle{clabel}=[fill=white, draw=none, shape=rectangle, tikzit fill={rgb,255: red,56; green,255; blue,242}, font={\footnotesize}, inner sep=1pt, tikzit category=labels]
\tikzstyle{empty diagram}=[draw={gray!40!white}, dashed, shape=rectangle, minimum width=1cm, minimum height=1cm, tikzit category=misc]
\tikzstyle{amap}=[fill=white, draw=black, shape=NEbox, tikzit category=asymmetric, tikzit fill=yellow, tikzit shape=rectangle]
\tikzstyle{amap conj}=[fill=white, draw=black, shape=NWbox, tikzit category=asymmetric, tikzit fill=green, tikzit shape=rectangle]
\tikzstyle{amap adj}=[fill=white, draw=black, shape=SEbox, tikzit category=asymmetric, tikzit fill=red, tikzit shape=rectangle]
\tikzstyle{amap trans}=[fill=white, draw=black, shape=SWbox, tikzit category=asymmetric, tikzit fill=orange, tikzit shape=rectangle]
\tikzstyle{astate}=[fill=white, draw=black, shape=NEtriangle, tikzit category=asymmetric, tikzit shape=circle, tikzit fill=yellow]
\tikzstyle{astate conj}=[fill=white, draw=black, shape=NWtriangle, tikzit category=asymmetric, tikzit shape=circle, tikzit fill=green]
\tikzstyle{astate adj}=[fill=white, draw=black, shape=SEtriangle, tikzit category=asymmetric, tikzit shape=circle, tikzit fill=red]
\tikzstyle{astate trans}=[fill=white, draw=black, shape=SWtriangle, tikzit category=asymmetric, tikzit shape=circle, tikzit fill=orange]
\tikzstyle{tri}=[fill={rgb,255: red,255; green,136; blue,136}, tikzit fill=red, draw=black, shape=regular polygon, regular polygon sides=3, regular polygon rotate=30, scale=0.7, inner sep=0.2pt, minimum size=7mm, tikzit shape=rectangle]
\tikzstyle{tri left}=[fill={rgb,255: red,255; green,136; blue,136}, tikzit fill=red, draw=black, shape=regular polygon, regular polygon sides=3, regular polygon rotate=210, scale=0.7, inner sep=0.2pt, minimum size=7mm, tikzit shape=rectangle]
\tikzstyle{new style 0}=[Z phase dot, fill={rgb,255: red,174; green,189; blue,255}, draw=black, shape=circle]
\tikzstyle{new style 1}=[fill=white, draw=black, shape=rectangle]
\tikzstyle{medium box}=[fill=white, draw=black, shape=rectangle, minimum width=0.75cm, minimum height=1.5cm]
\tikzstyle{cnot ctrl}=[fill=black, draw=black, shape=circle, inner sep=0pt, minimum width=1.2mm, tikzit category=circuit]
\tikzstyle{cnot targ}=[fill=white, draw=white, shape=circle, tikzit category=circuit, label={center:$\oplus$}, inner sep=0pt, minimum width=2.1mm, tikzit fill={rgb,255: red,102; green,204; blue,255}, tikzit draw=black]
\tikzstyle{my_ket}=[fill=white, draw=black, shape=regular polygon, regular polygon sides=3, regular polygon rotate=0, scale=0.7, inner sep=1pt, tikzit category=circuit, tikzit shape=rectangle, tikzit fill=green]
\tikzstyle{my_bra}=[fill=white, draw=black, shape=regular polygon, regular polygon sides=3, regular polygon rotate=180, scale=0.7, inner sep=1pt, tikzit category=circuit, tikzit shape=rectangle, tikzit fill=red]

\tikzstyle{smallTensor}=[fill={rgb,255: red,255; green,0; blue,255}, draw={rgb,150: red,255; green,0; blue,255}, shape=circle, minimum width=5mm, font={\footnotesize}, text height=1.5ex, text depth=0.25ex, tikzit category=zx]
\tikzstyle{bigTensor}=[fill={rgb,255: red,51; green,204; blue,204}, draw={rgb,150: red,51; green,204; blue,204}, shape=rectangle, inner sep=0.6mm, minimum height=10mm, minimum width=35mm, font={\footnotesize}, text height=1.5ex, text depth=0.25ex, tikzit category=zx]

\tikzstyle{smallMPO}=[fill={rgb,255: red,255; green,153; blue,51}, draw={rgb,150: red,255; green,153; blue,51}, shape=rectangle, inner sep=0.6mm, minimum height=10mm, minimum width=5mm, font={\footnotesize}, text height=1.5ex, text depth=0.25ex, tikzit category=zx]
\tikzstyle{bigMPO}=[fill={rgb,255: red,51; green,204; blue,51}, draw={rgb,150: red,51; green,204; blue,51}, shape=rectangle, inner sep=0.6mm, minimum height=10mm, minimum width=35mm, font={\footnotesize}, text height=1.5ex, text depth=0.25ex, tikzit category=zx]

\tikzstyle{rounded_green_box}=[
    shape=rectangle,
    rounded corners=2mm,
    fill={rgb,255: red,87; green,226; blue,67},
    fill opacity=0.2,
    text opacity = 1,
    minimum height=15mm,
    minimum width=60mm,
    tikzit category=circuit,
    inner xsep=2mm,
    align = left,
    text width = 56mm,
    font={\fontsize{10pt}{12pt}\selectfont}
]

\tikzstyle{rounded_green_box_shallow}=[
    shape=rectangle,
    rounded corners=2mm,
    fill={rgb,255: red,87; green,226; blue,67},
    fill opacity=0.2,
    text opacity = 1,
    minimum height=25mm,
    minimum width=40mm,
    tikzit category=circuit,
    inner xsep=2mm,
    align = left,
    text width = 36mm,
    font={\fontsize{10pt}{12pt}\selectfont}
]

\tikzstyle{tensor_0_orange}=[
    shape=rectangle,
    fill = {rgb,255: red,226; green,135; blue,67}, 
    text opacity = 1,
    minimum height=10mm,
    minimum width=5mm,
    inner xsep=1mm,
    tikzit category=circuit,
    font={\fontsize{10pt}{12pt}\selectfont}
]
\tikzstyle{tensor_5_orange}=[
    shape=rectangle,
    fill = {rgb,255: red,226; green,135; blue,67}, 
    text opacity = 1,
    minimum height=10mm,
    minimum width=35mm,
    tikzit category=circuit,
    font={\fontsize{10pt}{12pt}\selectfont}
]

\tikzstyle{tensor_0}=[
    shape=rectangle,
    rounded corners=2mm,
    fill = {rgb,255: red,51; green,204; blue,204}, 
    draw={rgb,150: red,51; green,204; blue,204},
    text opacity = 1,
    minimum height=10mm,
    minimum width=5mm,
    inner xsep=1mm,
    tikzit category=circuit,
    font={\fontsize{10pt}{12pt}\selectfont}
]

\tikzstyle{tensor_tight}=[
    shape=rectangle,
    rounded corners=2mm,
    fill = {rgb,255: red,51; green,204; blue,204}, 
    draw={rgb,150: red,51; green,204; blue,204},
    text opacity = 1,
    minimum height=10mm,
    minimum width=5mm,
    inner xsep=0mm,
    tikzit category=circuit,
    font={\fontsize{7pt}{12pt}\selectfont}
]

\tikzstyle{tensor_5}=[
    shape=rectangle,
    rounded corners=3mm,
    fill = {rgb,255: red,51; green,204; blue,204}, 
    draw={rgb,150: red,51; green,204; blue,204},
    text opacity = 1,
    minimum height=10mm,
    minimum width=35mm,
    tikzit category=circuit,
    font={\fontsize{10pt}{12pt}\selectfont}
]

\tikzstyle{tensor_4}=[
    shape=rectangle,
    rounded corners=3mm,
    fill = {rgb,255: red,51; green,204; blue,204}, 
    draw={rgb,150: red,51; green,204; blue,204},
    text opacity = 1,
    minimum height=10mm,
    minimum width=27.5mm,
    tikzit category=circuit,
    font={\fontsize{10pt}{12pt}\selectfont}
]

\tikzstyle{rounded_gray_box}=[
    shape=rectangle,
    rounded corners=3mm,
    fill={rgb,255: red,203; green,192; blue,225},
    fill opacity=0.2,
    text opacity = 1,
    align = left,
    text width = 36mm,
    inner xsep = 2mm,
    minimum height=15mm,
    minimum width=40mm,
    tikzit category=circuit,
    inner sep=2mm,
    font={\fontsize{10pt}{12pt}\selectfont}
]

\tikzstyle{gateTest_RU}=[shape=rectangle, fill=white, draw=black, minimum height=27.5mm, minimum width=7.5mm, font={\fontsize{10pt}{12pt}\selectfont}, tikzit category=circuit, inner sep = 5, align=center] 
\tikzstyle{gateTest_gray}=[shape=rectangle, dashed, fill=white, draw={rgb,255: red,203; green,192; blue,225}, minimum height=27.5mm, minimum width=7.5mm, font={\fontsize{10pt}{12pt}\selectfont}, tikzit category=circuit, inner sep = 2, align=center] 

\tikzstyle{gateTest_0}=[shape=rectangle, fill=white, draw=black, minimum height=5mm, minimum width=7.5mm, font={\fontsize{10pt}{12pt}\selectfont}, tikzit category=circuit, inner sep = 0] 
\tikzstyle{gateTest_1}=[shape=rectangle, align = center, fill=white, draw=black, minimum height=10.0mm, minimum width=7.5mm, font={\fontsize{10pt}{12pt}\selectfont}, tikzit category=circuit, inner sep = 0] 
\tikzstyle{gateTest_2}=[shape=rectangle, align = center, fill=white, draw=black, minimum height=17.5mm, minimum width=7.5mm, font={\fontsize{10pt}{12pt}\selectfont}, tikzit category=circuit, inner sep = 0] 
\tikzstyle{gateTest_3}=[shape=rectangle, align = center, fill=white, draw=black, minimum height=25.0mm, minimum width=7.5mm, font={\fontsize{10pt}{12pt}\selectfont}, tikzit category=circuit, inner sep = 0] 
\tikzstyle{gateTest_4}=[shape=rectangle, align = center, fill=white, draw=black, minimum height=32.5mm, minimum width=7.5mm, font={\fontsize{10pt}{12pt}\selectfont}, tikzit category=circuit, inner sep = 0] 
\tikzstyle{gateTest_5}=[shape=rectangle, align = center, fill=white, draw=black, minimum height=40.0mm, minimum width=7.5mm, font={\fontsize{10pt}{12pt}\selectfont}, tikzit category=circuit, inner sep = 0] 
\tikzstyle{gateTest_6}=[shape=rectangle, align = center, fill=white, draw=black, minimum height=47.5mm, minimum width=7.5mm, font={\fontsize{10pt}{12pt}\selectfont}, tikzit category=circuit, inner sep = 0] 
\tikzstyle{gateTest_7}=[shape=rectangle, fill=white, draw=black, minimum height=55.0mm, minimum width=7.5mm, font={\fontsize{10pt}{12pt}\selectfont}, tikzit category=circuit, inner sep = 0] 
\tikzstyle{gateTest_8}=[shape=rectangle, fill=white, draw=black, minimum height=62.5mm, minimum width=7.5mm, font={\fontsize{10pt}{12pt}\selectfont}, tikzit category=circuit, inner sep = 0] 
\tikzstyle{gateTest_9}=[shape=rectangle, fill=white, draw=black, minimum height=70.0mm, minimum width=7.5mm, font={\fontsize{10pt}{12pt}\selectfont}, tikzit category=circuit, inner sep = 0] 
\tikzstyle{gateTest_10}=[shape=rectangle, fill=white, draw=black, minimum height=77.5mm, minimum width=7.5mm, font={\fontsize{10pt}{12pt}\selectfont}, tikzit category=circuit, inner sep = 0] 
\tikzstyle{gateTest_11}=[shape=rectangle, fill=white, draw=black, minimum height=85.0mm, minimum width=7.5mm, font={\fontsize{10pt}{12pt}\selectfont}, tikzit category=circuit, inner sep = 0] 
\tikzstyle{gateTest_12}=[shape=rectangle, fill=white, draw=black, minimum height=92.5mm, minimum width=7.5mm, font={\fontsize{10pt}{12pt}\selectfont}, tikzit category=circuit, inner sep = 0] 
\tikzstyle{gateTest_13}=[shape=rectangle, fill=white, draw=black, minimum height=100.0mm, minimum width=7.5mm, font={\fontsize{10pt}{12pt}\selectfont}, tikzit category=circuit, inner sep = 0] 
\tikzstyle{gateTest_14}=[shape=rectangle, fill=white, draw=black, minimum height=107.5mm, minimum width=7.5mm, font={\fontsize{10pt}{12pt}\selectfont}, tikzit category=circuit, inner sep = 0] 
\tikzstyle{gateTest_15}=[shape=rectangle, fill=white, draw=black, minimum height=115.0mm, minimum width=7.5mm, font={\fontsize{10pt}{12pt}\selectfont}, tikzit category=circuit, inner sep = 0] 
\tikzstyle{gateTest_16}=[shape=rectangle, fill=white, draw=black, minimum height=122.5mm, minimum width=7.5mm, font={\fontsize{10pt}{12pt}\selectfont}, tikzit category=circuit, inner sep = 0] 
\tikzstyle{gateTest_17}=[shape=rectangle, fill=white, draw=black, minimum height=130.0mm, minimum width=7.5mm, font={\fontsize{10pt}{12pt}\selectfont}, tikzit category=circuit, inner sep = 0] 
\tikzstyle{gateTest_18}=[shape=rectangle, fill=white, draw=black, minimum height=137.5mm, minimum width=7.5mm, font={\fontsize{10pt}{12pt}\selectfont}, tikzit category=circuit, inner sep = 0] 
\tikzstyle{gateTest_19}=[shape=rectangle, fill=white, draw=black, minimum height=145.0mm, minimum width=7.5mm, font={\fontsize{10pt}{12pt}\selectfont}, tikzit category=circuit, inner sep = 0]

\tikzstyle{measurement}=[
    shape=rectangle, 
    fill=white, 
    draw=black, 
    minimum height=5mm, 
    minimum width=7.5mm, 
    tikzit category=circuit,
    path picture={
        \draw[line width = 0.6pt] (path picture bounding box.center) ++(-2mm,-1mm) arc (180:0:2mm);
        \draw[->,line width = 0.6pt] (path picture bounding box.center) ++(0,-0.5mm) -- ++(2mm,2mm);
    }
]

\tikzstyle{pink}=[-, dashed, dash pattern=on 2pt off 0.5pt, thick, draw={rgb,255: red,255; green,20; blue,147}]
\tikzstyle{green}=[-, dashed, dash pattern=on 2pt off 0.5pt, thick, draw={rgb,255: red,87; green,226; blue,67}]
\tikzstyle{orange}=[-, dashed, dash pattern=on 2pt off 0.5pt, thick, draw={rgb,255: red,226; green,135; blue,67}]
\tikzstyle{hadamard edge}=[-, dashed, dash pattern=on 2pt off 0.5pt, thick, draw={rgb,255: red,68; green,136; blue,255}]
\tikzstyle{box edge}=[-, dashed, dash pattern=on 2pt off 0.5pt, thick, draw={rgb,255: red,203; green,192; blue,225}]
\tikzstyle{brace edge}=[-, tikzit draw=blue, decorate, decoration={brace,amplitude=1mm,raise=-1mm}]
\tikzstyle{physicalEdge}=[-, tikzit draw={rgb,255: red,255; green,0; blue,255}, draw={rgb,150: red,255; green,0; blue,255}]
\tikzstyle{physicalEdgeThick}=[-, tikzit draw={rgb,255: red,255; green,0; blue,255}, draw={rgb,150: red,255; green,0; blue,255}, line width=5pt]

\tikzstyle{virtualEdgeOrange}=[-, tikzit draw={rgb,255: red,226; green,135; blue,67}, draw={rgb,255: red,226; green,135; blue,67}, line width=1pt]
\tikzstyle{virtualEdge}=[-, tikzit draw={rgb,150: red,51; green,204; blue,204}, draw={rgb,255: red,51; green,204; blue,204}, line width=1pt]
\tikzstyle{virtualEdge1}=[-, tikzit draw={rgb,150: red,51; green,204; blue,204}, draw={rgb,255: red,51; green,204; blue,204}, line width=1pt]
\tikzstyle{virtualEdge3}=[-, tikzit draw={rgb,150: red,51; green,204; blue,204}, draw={rgb,255: red,51; green,204; blue,204}, line width=3pt]
\tikzstyle{virtualEdge5}=[-, tikzit draw={rgb,150: red,51; green,204; blue,204}, draw={rgb,255: red,51; green,204; blue,204}, line width=5pt]
\tikzstyle{diredge}=[->]
\tikzstyle{double edge}=[-, double, shorten <=-1mm, shorten >=-1mm, double distance=2pt]
\tikzstyle{gray edge}=[-, {gray!60!white}]
\tikzstyle{pointer edge}=[->, very thick, gray]
\tikzstyle{boldedge}=[-, line width=1.6pt, shorten <=-0.17mm, shorten >=-0.17mm]
\tikzstyle{bidir edge}=[<->, very thick, draw={rgb,255: red,191; green,191; blue,191}]

\input{tikzit/quantum.tikzdefs}

\usepackage{amsmath,amssymb,amsthm,mathrsfs,amsfonts,dsfont,mathtools,physics}
\usepackage{braket}

\usepackage{ragged2e}
\usepackage{graphicx}
\usepackage{hyperref}
\usepackage[capitalize]{cleveref}

\usepackage{float}

\usepackage{microtype}

\usepackage[caption=false]{subfig}
\usepackage{soul}
\usepackage{adjustbox}

\usepackage{bm}
\usepackage{enumerate}
\usepackage{color}

\usepackage{appendix}

\usepackage{enumitem}
\usepackage[normalem]{ulem}
\usepackage{comment}
\usepackage{xcolor}

\usepackage{algorithm}       
\usepackage{algpseudocode}

\usepackage{amsthm}

\newtheorem{theorem}{Theorem}
\newtheorem{lemma}{Lemma}

\newtheorem{corollary}{Corollary}

\Crefname{theorem}{Theorem}{Theorems}
\theoremstyle{remark}

\newcommand{\qmaddress}{\affiliation{Quantum Motion, 9 Sterling Way, London N7 9HJ, United Kingdom}}
\newcommand{\oxaddress}{\affiliation{Mathematical Institute, University of Oxford, Woodstock Road, Oxford OX2 6GG, United Kingdom}}
\newcommand{\mrnaaddress}{\affiliation{Moderna, Cambridge, MA 02139, USA}}

\begin{document}

\title{A Shadow Enhanced Greedy Quantum Eigensolver}

\author{Jona Erle}
\email{jona.erle@balliol.ox.ac.uk}
\oxaddress
\qmaddress
\mrnaaddress

\author{B\'alint Koczor}
\email{balint.koczor@maths.ox.ac.uk}
\oxaddress
\qmaddress

\begin{abstract}
While ground-state preparation is expected to be a primary application of quantum computers, it is also an essential subroutine for many fault-tolerant algorithms. 
In early fault-tolerant regimes, logical measurements remain costly, motivating adaptive, shot-frugal state-preparation strategies that efficiently utilize each measurement.
We introduce the Shadow Enhanced Greedy Quantum Eigensolver (SEGQE) as a greedy, shadow-assisted framework for measurement-efficient ground-state preparation. SEGQE uses classical shadows to evaluate---in parallel and entirely in classical post-processing---the energy reduction induced by large collections of local candidate gates, greedily selecting at each step the gate with the largest estimated energy decrease. 
We derive rigorous worst-case per-iteration sample-complexity bounds for SEGQE, exhibiting logarithmic dependence on the number of candidate gates. Numerical benchmarks on finite transverse-field Ising models and ensembles of random local Hamiltonians demonstrate convergence in a number of iterations that scales approximately linearly with system size, while maintaining high-fidelity ground-state approximations and competitive energy estimates.
Together, our empirical scaling laws and rigorous per-iteration guarantees establish SEGQE as a measurement-efficient state-preparation primitive well suited to early fault-tolerant quantum computing architectures.
\end{abstract}

\maketitle
\section{Introduction}
Preparing high-fidelity ground states is a central task in quantum computing, with applications in quantum chemistry \cite{kandala2017hardware}, quantum optimization \cite{farhi2014quantum, moll2018quantum}, and quantum machine learning \cite{farhi2018classification, cong2019quantum, peruzzo2014variational, cerezo2021variational, koczor2022quantum, koczor2022quantum2}. Furthermore, the efficiency of many fault-tolerant algorithms (e.g. phase estimation \cite{wang2022state}) hinges on the quality of their input state.

The Variational Quantum Eigensolver (VQE), a prominent example of a variational quantum algorithm (VQA), is a well-studied, hybrid, quantum-classical algorithm proposed to prepare the ground state of a Hamiltonian $H$ by variationally optimizing the parameters of an ansatz circuit $U(\theta)$ \cite{cerezo2021variational, peruzzo2014variational}. Since the ansatz structure is typically fixed and only the parameters $\theta$ are varied during training, choosing the right ansatz structure is crucial as it determines the family of expressible states \cite{kandala2017hardware, romero2018strategies, wecker2015progress}. It has been observed in the literature that fixed ansatz structures lead to a fundamental trade-off between expressivity and trainability \cite{holmes2022connecting}: highly expressive ansätze tend to exhibit barren plateaus of exponentially small gradients \cite{mcclean2018barren, larocca2024review, cerezo2021cost, wang2021noise}, whereas simpler ansätze may lack the ability to approximate the target state accurately. Trainability issues are further exacerbated by the presence of poor local minima \cite{anschuetz2022quantum, bittel2021training}. 
Consequently, generic applications of both gradient-based and gradient-free VQAs may demand a prohibitively large measurement budget, constituting a major bottleneck for scalability on realistic hardware. 

Resource minimization becomes even more critical in early fault-tolerant regimes, where deeper circuits may be feasible but logical measurement rates remain limited \cite{zimboras2025myths}. These considerations have motivated the development of alternative state-preparation strategies that retain flexibility, avoid expensive gradient estimation, and aim to extract maximal information from each quantum measurement \cite{boyd2022training, hwang2025preparing, puig2026warm,grimsley2019adaptive, nakanishi2020sequential}. 
Classical shadows \cite{huang2020predicting} are particularly relevant in this context, as they enable the simultaneous estimation of many expectation values from a modest number of measurements. Specifically, given a quantum state $\rho$, one performs a measurement in a randomly chosen basis (e.g., random local Pauli or Clifford measurements) and constructs from each outcome a classical snapshot (a shadow) $\hat{\rho}$. These snapshots can be used to estimate expectation values of a wide range of observables. Remarkably, for fixed accuracy and failure probability, the number of required snapshots scales only logarithmically in the number $M$ of target observables, while generally growing exponentially with their locality under random local Pauli measurements. 

In this work, we introduce a resource-efficient, greedy, adaptive ground-state preparation algorithm that we term \emph{Shadow Enhanced Greedy Quantum Eigensolver} (SEGQE). Given an input state $\ket{\psi_0}$ and a Hamiltonian $H$, SEGQE iteratively constructs a circuit starting from the identity. At each step, a batch of classical shadows of the current state is collected and used to estimate the energy decrease induced by a set of candidate gate additions. The gate yielding the largest predicted energy decrease is appended to the circuit. Repeating this procedure progressively refines the circuit and yields a state that approximates the ground state of $H$. 

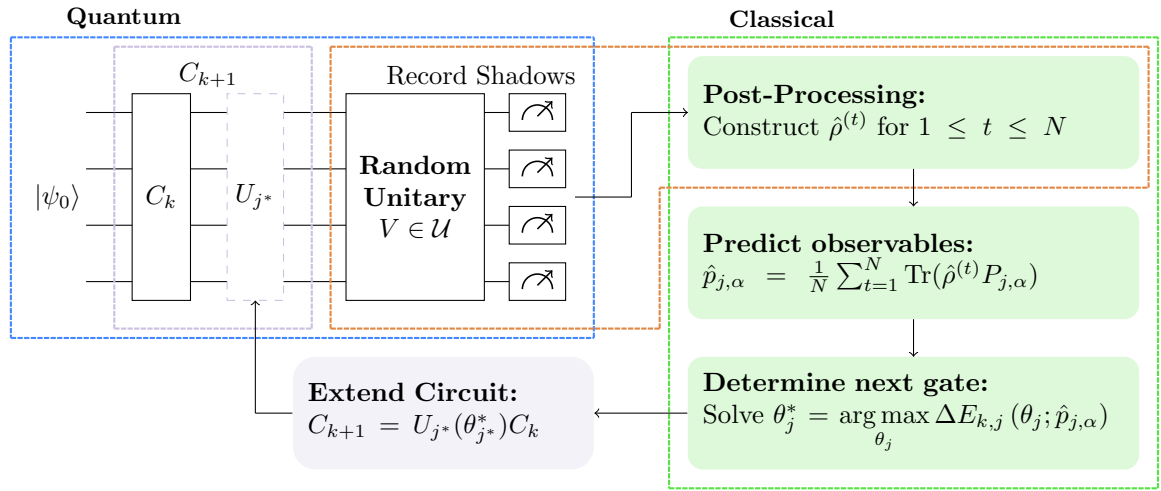
\begin{figure*}[htbp]
	\begin{center}
		\begin{tikzpicture}
	\begin{pgfonlayer}{nodelayer}
		\node [style=none] (0) at (-4.5, 2.5) {};
		\node [style=none] (1) at (11, 2.5) {};
		\node [style=none] (2) at (-2.5, 0.5) {};
		\node [style=none] (3) at (-2.5, -1) {};
		\node [style=none] (4) at (-2.5, -2.5) {};
		\node [style=none] (5) at (-2.5, -4) {};
		\node [style=none] (7) at (9.5, 0.5) {};
		\node [style=none] (8) at (9.5, -1) {};
		\node [style=none] (9) at (9.5, -2.5) {};
		\node [style=none] (10) at (9.5, -4) {};
		\node [style={gateTest_RU}] (11) at (6.25, -1.75) {\textbf{Random}\\[-0ex] \textbf{Unitary}\\[-0ex] $V\in\mathcal{U}$};
		\node [style=measurement] (12) at (9.5, 0.5) {};
		\node [style=measurement] (13) at (9.5, -1) {};
		\node [style=measurement] (14) at (9.5, -2.5) {};
		\node [style=measurement] (15) at (9.5, -4) {};
		\node [style={gateTest_RU}] (17) at (-0.5, -1.75) {$C_k$};
		\node [style=none] (18) at (-1.5, 3) {\textbf{Quantum}};
		\node [style=none] (19) at (13, 2.5) {};
		\node [style=none] (20) at (11, -5.5) {};
		\node [style=none] (21) at (13, -9.5) {};
		\node [style=none] (22) at (-4.5, -5.5) {};
		\node [style=none] (23) at (4, 2.25) {};
		\node [style=none] (24) at (4, -5.25) {};
		\node [style=none] (25) at (25.75, -1.5) {};
		\node [style=none] (26) at (25.75, 2.25) {};
		\node [style=none] (28) at (26, 2.5) {};
		\node [style=none] (29) at (26, -9.5) {};
		\node [style=none] (30) at (16, 3) {\textbf{Classical}};
		\node [style=none] (39) at (-1.75, 2.25) {};
		\node [style=none] (40) at (3.5, 2.25) {};
		\node [style=none] (41) at (-1.75, -5.25) {};
		\node [style=none] (42) at (3.5, -5.25) {};
		\node [style={gateTest_gray}] (44) at (2, -1.75) {$U_{j^*}$};
		\node [style=none] (46) at (12.75, -1.5) {};
		\node [style=none] (47) at (12.75, -5.25) {};
		\node [style=none] (48) at (10.5, -1.75) {};
		\node [style=none] (49) at (12, -1.75) {};
		\node [style=none] (50) at (12, 0.5) {};
		\node [style=none] (51) at (13.5, 0.5) {};
		\node [style={rounded_green_box}] (52) at (19.5, -7.5) {\textbf{Determine next gate:}\\[0ex] Solve $\theta^*_{j} = \underset{\theta_j}{\arg\max} \,\Delta E_{k,j}\left(\theta_j; \hat{p}_{j,\alpha}\right)$};
		\node [style={rounded_green_box}] (54) at (19.5, -3.5) {\textbf{Predict observables:}\\[-0ex] $\hat{p}_{j,\alpha} = \frac{1}{N}\sum_{t=1}^N\text{Tr}(\hat{\rho}^{(t)}P_{j,\alpha})$}; 
		\node [style={rounded_green_box}] (55) at (19.5, 0.5) {\textbf{Post-Processing:}\\[-0ex] Construct $\hat{\rho}^{(t)}$ for $1\leq t \leq N$};
		\node [style=none] (56) at (19.5, -1) {};
		\node [style=none] (57) at (19.5, -2) {};
		\node [style=none] (58) at (19.5, -5) {};
		\node [style=none] (59) at (19.5, -6) {};
		\node [style=none] (60) at (2, -7.5) {};
		\node [style=none] (61) at (2, -4.5) {};
		\node [style=Label] (67) at (-3.25, -1.75) {$\ket{\psi_0}$};
		\node [style=Label] (68) at (8, 1.5) {Record Shadows};
		\node [style=Label] (69) at (0.75, 1.5) {$C_{k+1}$};
		\node [style={rounded_gray_box}] (70) at (7, -7.5) {\textbf{Extend Circuit:}\\[-0ex] $C_{k+1} = U_{j^*}(\theta_{j^*}^*)C_k$};
	\end{pgfonlayer}
	\begin{pgfonlayer}{edgelayer}
		\draw (2.center) to (7.center);
		\draw (3.center) to (8.center);
		\draw (4.center) to (9.center);
		\draw (5.center) to (10.center);
		\draw [style=hadamard edge] (0.center) to (1.center);
		\draw [style=hadamard edge] (20.center) to (1.center);
		\draw [style=hadamard edge] (0.center) to (22.center);
		\draw [style=hadamard edge] (22.center) to (20.center);
		\draw [style=box edge] (39.center) to (40.center);
		\draw [style=box edge] (39.center) to (41.center);
		\draw [style=box edge] (41.center) to (42.center);
		\draw [style=box edge] (42.center) to (40.center);
		\draw [style=green] (19.center) to (21.center);
		\draw [style=green] (21.center) to (29.center);
		\draw [style=green] (29.center) to (28.center);
		\draw [style=green] (28.center) to (19.center);
		\draw [style=orange] (23.center) to (24.center);
		\draw [style=orange] (25.center) to (26.center);
		\draw [style=orange] (26.center) to (23.center);
		\draw [style=orange] (24.center) to (47.center);
		\draw [style=orange] (47.center) to (46.center);
		\draw [style=orange] (46.center) to (25.center);
		\draw (48.center) to (49.center);
		\draw (49.center) to (50.center);
		\draw [style=diredge] (50.center) to (51.center);
		\draw [style=diredge] (56.center) to (57.center);
		\draw [style=diredge] (58.center) to (59.center);
		\draw [style=diredge] (60.center) to (61.center);
		\draw (70) to (60.center);
		\draw [style=diredge] (52) to (70);
	\end{pgfonlayer}
\end{tikzpicture}
	\end{center}
	\caption{Schematic overview of the Shadow-Enhanced Greedy Quantum Eigensolver (SEGQE). At iteration $k$, a quantum computer prepares the state $\ket{\psi_k}=C_k\ket{\psi_0}$ and $N$ independent classical shadows are collected by applying random unitaries $V\in \mathcal{U}$ and measuring in the computational basis. In this work, $\mathcal{U}=\{I, H, S^\dagger H\}^{\otimes n}$, corresponding to uniformly random single-qubit Pauli measurements. The shadows are then used to estimate all Pauli expectation values needed to evaluate, for each candidate gate $U_j(\theta)\in\mathcal{G}$, the energy decrease $\Delta E_{k,j}(\theta)$. A classical optimizer finds $(j^*, \theta^*_{j^*})$ maximizing the expected decrease and appends $U_{j^*}(\theta^*_{j^*})$ to obtain $C_{k+1}$. The procedure repeats until convergence or until a predefined circuit depth is reached.}
	\label{fig:SEGQE}
\end{figure*}

The central advantage of SEGQE lies in its measurement efficiency. At each iteration, a single batch of classical shadows can be reused to evaluate a large set of candidate gate additions entirely in classical post-processing, yielding a per-iteration sample complexity that scales only logarithmically with the size of the candidate set. Furthermore, this design offloads the computationally expensive search over gate additions to classical hardware, enabling efficient parallelization and utilization of high-performance computing resources. While the circuit depth increases with successive gate additions, the per-iteration measurement cost grows only mildly with system size, and the information extracted per shot retains the asymptotic optimality guarantees of classical shadows. This feature is particularly well suited to early fault-tolerant regimes, where deeper circuits than typical VQA ansätze may be feasible, but relatively slow logical measurement rates demand shot-frugal techniques. In such settings, minimizing the measurement budget while exploiting classical computational power is crucial. SEGQE therefore provides a promising approach for applications such as initial state preparation for fault-tolerant phase estimation. 

A variety of iterative state-preparation methods have been proposed in the literature~\cite{motta2020determining, grimsley2019adaptive, feniou2025greedy} (see \Cref{appendix:related_works}). Here, we introduce a general classical-shadows-based framework that accommodates large sets of arbitrary local unitary gate additions and admits rigorous worst-case bounds on the per-iteration sample complexity with logarithmic dependence on the number of candidate gates. Although greedy strategies are not guaranteed to converge to the exact ground state in general, our numerical experiments demonstrate that SEGQE prepares high-fidelity ground-state approximations across a range of Hamiltonians.

The remainder of this article is structured as follows. In \cref{sec:procedure} we give a detailed description of the SEGQE algorithm, and in \cref{sec:guarantees} we provide rigorous worst-case upper bounds on its per-iteration sample complexity. In \cref{sec:numerics} we benchmark SEGQE on transverse-field Ising models and random local Hamiltonians, and study how the choice of gate set impacts convergence and final accuracy. Finally, in \cref{sec:conclusion} we reflect on our results and provide directions for future research.

\section{The Shadow Enhanced Greedy Quantum Eigensolver}
\label{sec:main}
\subsection{Procedure}
\label{sec:procedure}
Starting from an empty circuit $C_0=I$, the Shadow Enhanced Greedy Quantum Eigensolver (SEGQE) iteratively constructs an approximation to the ground state of a given $n$-qubit Hamiltonian $H=\sum^r_{i=1}c_iP_i$ by growing a quantum circuit $C$. We denote by $C_k$ the circuit after $k$ iterations. Applied to an initial state $\ket{\psi_0}$, which may encode prior classical information such as a Hartree--Fock state or a matrix product state, the circuit prepares $\ket{\psi_k}=C_k\ket{\psi_0}$. At each iteration, we record a collection of $N$ independent classical shadows obtained from uniformly random local Pauli measurements of $\ket{\psi_k}$, which allows us to estimate the energy difference
\begin{equation}
   \Delta E_{k, j}(\theta) =\bra{\psi_k}H\ket{\psi_k}-\bra{\psi_k}U^\dagger _j(\theta)HU_j(\theta)\ket{\psi_k},
\end{equation}
for every candidate gate $U_j(\theta)$ in a given set $\mathcal{G}=\{U_j(\theta)\}^K_{j=1}$. Classically maximizing the estimated energy difference over $\theta$ for all candidate gates enables us to identify the energy-minimizing gate $U_{j^*}(\theta^*_{j^*})$ and update the circuit as $C_k\to C_{k+1}=U_{j^*}(\theta^*_{j^*})C_k$. The procedure is repeated until either a predefined circuit depth $D$ is reached or no candidate gate yields an energy reduction exceeding a predefined threshold $\Delta >0$. 

To understand how SEGQE estimates the energy-minimizing parameters $\theta^*_j$ for each gate in $\mathcal{G}$ at iteration $k$ (for notational convenience, we drop the index $k$ in the following), consider a single arbitrarily parametrized $n$-qubit unitary $U(\theta) \in \mathrm{SU}(2^n)$ of locality $m\leq n$. 
Denote by $Q^u\subseteq\{1,\ldots n\}$ the qubit support of $U(\theta)$, and by $Q^c$ its complement, on which $U(\theta)$ acts trivially. Likewise, let $Q^i\subseteq\{1,\ldots n\}$ denote the support of the Hamiltonian term $P_i$. We can decompose each Pauli $P_i$ with respect to the bipartition $Q^u\cup Q^c$ as $P_i=P^u_i\otimes P^c_i$, where $P^u_i$ acts on $Q^u$ and $P^c_i$ acts on $Q^c$. Furthermore, we define the set $I^u=\left\{i\in\{1,\ldots,r\}: Q^i\cap Q^u\neq \emptyset\right\}$ as the indices of the Hamiltonian terms which act nontrivially on the qubits in $Q^u$. Hamiltonian terms with $Q^i \cap Q^u = \emptyset$ commute with $U(\theta)$ and therefore do not contribute to the energy difference. The energy difference induced by $U(\theta)$ is then given by
\begin{align}
    \Delta E(\theta)=&\bra{\psi}H\ket{\psi}-\bra{\psi}U^\dagger(\theta) H U(\theta) \ket{\psi}\\
	=& \sum_{i\in I^u}c_i 
	\left(\bra{\psi}P_i\ket{\psi}{-}\bra{\psi}U^\dagger(\theta) P^u_i U(\theta)\otimes P_i^c\ket{\psi} \right) \nonumber
\end{align}
    Since the $m$-qubit Pauli operators $\{P_\mu^u\}_{\mu=0}^{4^m-1}$ form an orthonormal basis (with respect to the Hilbert--Schmidt inner product) on the operator space corresponding to the qubits in $Q^u$, we can expand
    \begin{align}
        U^\dagger(\theta) P^u_i U(\theta) &= \sum_{\mu=0}^{4^m-1}\frac{1}{2^m}\mathrm{tr}\left[U^\dagger(\theta) P^u_iU(\theta)P_\mu^u\right]P_\mu^u\\&=\sum_{\mu=0}^{4^m-1}r_{i\mu}P_\mu^u,
    \end{align}
    where $r_{i\mu}(\theta):=\frac{1}{2^m}\mathrm{tr}\left[U^\dagger(\theta) P^u_iU(\theta)P_\mu^u\right]$. Substituting this into the above expression for the energy difference yields
    \begin{equation}
        \label{eq:energy_differences}
        \Delta E(\theta)=\sum_{i\in I^u}\sum^{4^m-1}_{\mu=0}f_{i\mu}(\theta)p_{i\mu},
    \end{equation}
    where we define
    \begin{equation*}
        f_{i\mu}(\theta):= c_i(\delta_{P^u_i,P_\mu^u} - r_{i\mu}(\theta)),\quad p_{i\mu} := \bra{\psi} P_\mu^u \otimes P_i^c \ket{\psi}.
    \end{equation*}
    Notably, the expectation values $p_{i\mu}$ are independent of $\theta$, so once estimated, $\Delta E(\theta)$ can be constructed via \cref{eq:energy_differences} and thus evaluated and maximized entirely classically. To simplify notation, we define the composite index $\alpha:=(i,\mu)$. For each candidate gate $U_j(\theta)\in\mathcal{G}$, we define the associated Pauli operators $P_{j,\alpha}:=P^{u(j)}_\mu \otimes P^{c(j)}_i$, and their expectation values $p_{j,\alpha}:=\bra{\psi}P_{j,\alpha}\ket{\psi}$, where the superscripts $u(j)$ and $c(j)$ indicate the bipartition induced by the support of $U_j$. The corresponding energy differences can then be written as
    \begin{equation}
        \Delta E_j(\theta) =\sum_{\alpha\in \mathcal{F}_j}f_{j,\alpha}(\theta)p_{j,\alpha}, 
    \end{equation}
    where $\mathcal{F}_j$ denotes the set of indices $\alpha$ that appear in the expansion associated with gate $U_j$. 
    In the following section, we derive rigorous guarantees on the estimation accuracy achievable when using classical shadows to predict these energy differences. The overall procedure of SEGQE is summarized in \cref{alg:segqe}.

\begin{figure}[htbp]
    \begin{minipage}{1\linewidth}
        \begin{algorithm}[H]
  \caption{Shadow Enhanced Greedy Quantum Eigensolver (SEGQE)}
  \label{alg:segqe}
  \begin{algorithmic}[1]
    \Require Hamiltonian $H = \sum_{i=1}^r c_i P_i$, initial state $\ket{\psi_0}$,
             gate set $\mathcal{G} = \{U_j(\theta)\}_{j=1}^K$, 
             maximum depth $D$, threshold $\Delta > 0$
    \State $C_0 \gets I$, $k \gets 0$
    \While{$k < D$}
      \State $\ket{\psi_k} \gets C_k \ket{\psi_0}$
      \State Record $N$ independent shadows $\{\hat{\rho}_k^{(t)}\}_{t=1}^N$ of $\ket{\psi_k}$
      \State Use shadows to estimate all required $p_{j,\alpha}$ 
      \State $\Delta E_{\mathrm{max}} \gets 0$, $(j^*, \theta^*) \gets \text{None}$
      \ForAll{candidate gates $U_j(\theta) \in \mathcal{G}$}
        \State Use $p_{j,\alpha}$ to classically find $\theta^*_j=\underset{\theta}{\arg\max} \Delta E_{k,j}(\theta)$
        \If{$\Delta E_{k,j}(\theta^*_j) > \Delta E_{\mathrm{max}}$}
          \State $\Delta E_{\mathrm{max}} \gets \Delta E_{k,j}(\theta^*_j)$
          \State $(j^*, \theta^*) \gets (j, \theta^*_j)$
        \EndIf
      \EndFor
      \If{$\Delta E_{\mathrm{max}} \le \Delta$}
        \State \textbf{break}  
      \EndIf
      \State $C_{k+1} \gets U_{j^*}(\theta^*) C_k$
      \State $k \gets k + 1$
    \EndWhile
    \State \Return $C_k$
  \end{algorithmic}
\end{algorithm}
    \end{minipage}

\end{figure}

\subsection{Per-Iteration Guarantees}
\label{sec:guarantees}
Throughout this section, we consider an $n$-qubit Hamiltonian $H=\sum_{i=1}^r c_iP_i$, where each Pauli operator $P_i$ has locality at most $l$, and we define $c_\mathrm{max}:=\max_i|c_i|$. Furthermore, for a set $S$ of $n$-qubit operators, we denote by $M(S)$ the maximal number of Hamiltonian terms whose support overlaps nontrivially with the support of any single operator in $S$. If the locality of the operators in $S$ is bounded, $M(S)$ is closely related to the sparsity of $H$. For a state $\ket{\psi}$ and a unitary $U$, we define the energy differences
\begin{equation}
    \Delta E_{\psi}(U):=\bra{\psi}H\ket{\psi}-\bra{\psi}U^\dagger HU\ket{\psi}.
\end{equation}
As described in \cref{sec:procedure}, at iteration $k$ for a given set of candidate gates $\mathcal{G}=\{U_j(\theta)\}^K_{j=1}$, SEGQE maximizes the energy differences $\Delta E_{\psi_k}\left(U_j(\theta)\right)$ for all gates $U_j(\theta)\in\mathcal{G}$ over $\theta$ and compares their values at their respective optimal parameter values. 
\begin{theorem}
    \label{theorem1}
    Let $\ket{\psi}$ be an arbitrary $n$-qubit state and $\mathcal{G}=\{U_j\}_{j=1}^K$ a finite collection of unitaries with locality at most $m\leq n$. Define $M := M(\mathcal{G})$. Then, for any $\epsilon, \delta \in (0,1]$, a collection of
    \begin{equation}
        N =\frac{32\log\left(2K/\delta\right)}{9\epsilon^2}4^m3^{l}M^2c_\mathrm{max}^2
    \end{equation}
    independent classical shadows of $\ket{\psi}$ obtained from uniformly random local Pauli measurements is sufficient to simultaneously estimate, up to additive error at most $\epsilon$ and failure probability at most $\delta$, the energy differences $\Delta E_{\psi}(U_j)$ for all $j=1,...,K$.
\end{theorem}
\Cref{theorem1} guarantees that once a finite set of unitaries $\{U_j(\theta^*_j)\}_{j=1}^K$ has been specified, their associated energy differences can be compared with high confidence. In particular, it enables reliable verification that the selected gate yields a genuine energy decrease, provided $\Delta E_{\psi_k}(U_{j^*}(\theta^*_{j^*}))>\epsilon$.
Notably, however, \cref{theorem1} does not provide guarantees for the uniform evaluation of a continuous parameter space, but applies only to a collection of fixed unitaries. Therefore, if one wishes to obtain guarantees for all parameter values queried during the optimization procedure---which may be necessary to identify a good maximizer---one must treat each parameter setting as a distinct unitary. If the classical optimizer requires at most $T$ function evaluations per gate, the total number of queried unitaries is at most $KT$, leading to a worst-case measurement cost scaling logarithmically in $KT$. In practice, one may impose a maximum number of per-iteration, per-gate function evaluations $T_\mathrm{max}$ as an explicit input parameter to SEGQE, thereby allocating the classical optimizer a fixed evaluation budget. 

Maximizing the energy differences associated with arbitrarily parametrized unitaries can be challenging and may require many function evaluations. In contrast, for unitaries of the form $U_j(\theta)=e^{-i\frac{\theta}{2}X_j }$ with $X_j^2=I$, as is the case for Pauli rotations commonly used in practice, the energy difference $\Delta E_\psi(U_j(\theta))$ admits an analytical determination of the maximizer. Such involutory generators are therefore particularly well suited for SEGQE. \Cref{theorem2} provides an upper bound on the measurement cost required to accurately estimate these global maxima and compute corresponding maximizing parameters for a finite collection of such unitaries, thereby yielding an explicit per-iteration sample complexity bound in this setting.

\begin{theorem}
    \label{theorem2}
    Given $\ket{\psi}$ as an arbitrary $n$-qubit state, we assume $\mathcal{S}=\{X_j\}_{j=1}^K$ is a finite collection of Hermitian gate generators
    with locality at most $m\leq n$ satisfying $X_j^2=I$, which define the unitaries $U_j(\theta):=\exp\left(-i\frac{\theta}{2}X_j\right)$.
    Let $M:=M(\mathcal{S})$. Then, for any $\epsilon, \delta \in (0,1]$, a collection of
    \begin{equation}
        N =\frac{\left(3+2\sqrt{2}\right)8\log\left(4K/\delta\right)}{9\epsilon^2}4^m3^{l}M^2c_\mathrm{max}^2
    \end{equation}
    independent classical shadows of $\ket{\psi}$ obtained from uniformly random local Pauli measurements is sufficient to simultaneously estimate, up to additive error at most $\epsilon$ and failure probability at most $\delta$, the maximal values of the energy differences $\Delta E_\psi(U_j)(\theta)$ for all $j=1,\ldots,K$. Moreover, from the same data one can compute estimators $\bar{\theta}^*_j\in [0,2\pi)$ of the maximizers $\theta^*_j\in [0,2\pi)$ such that, with probability at least $1-\delta$, the energy decrease achieved by applying $U_j(\bar{\theta}^*_j)$ to $\ket{\psi}$ deviates by at most $\epsilon$ from the estimated maximum and by at most $2\epsilon$ from the true maximum.
\end{theorem}

Pauli operators $P\in\mathcal{P}^n$ satisfy $P^2=I$ and therefore fall within the class of generators considered in \cref{theorem2}. In particular, when the
gate set considered by SEGQE consists of all parameterized Pauli rotations with locality at most $m$, the resulting per-iteration sample complexity is given by the following corollary.
\begin{corollary}
    \label{cor:pauli-rotations}
    In the special case of \cref{theorem2} where the gate generators are given by
    $\mathcal{S}:= \{P\in \mathcal{P}^n:\text{$P$ has locality at most $m$}\}$, that is, all Pauli operators on $n$ qubits with locality at most $m$,
    for any $\epsilon, \delta \in (0,1]$, a collection of 
    \begin{equation}
        N =\frac{\left(3+2\sqrt{2}\right)8\log\left(4^{m+1}\binom{n}{m}/\delta\right)}{9\epsilon^2}3^{m+l}M^2c_\mathrm{max}^2
    \end{equation}
    independent classical shadows of $\ket{\psi}$ obtained from uniformly random local Pauli measurements is sufficient to guarantee the conclusions of \cref{theorem2}.
\end{corollary}
Both \cref{theorem1} and \cref{theorem2}, as well as \cref{cor:pauli-rotations}, provide worst-case upper bounds on the per-iteration sample complexity of SEGQE. In all cases, the measurement cost scales exponentially in both the locality $m$ of the candidate gates and the locality $l$ of the Hamiltonian terms, while depending only logarithmically on the number of considered gates. The exponential dependence on the Hamiltonian locality $l$ restricts SEGQE to local Hamiltonians, and the exponential dependence on $m$ restricts the practicality of SEGQE to local gate sets. On the other hand, the logarithmic dependence on the size of the gate sets implies that considering many local gates incurs only a mild measurement overhead. Furthermore, for fixed gate locality, the number of candidate gates grows at most polynomially with system size $n$. As a result, the per-iteration sample complexity depends on $n$ only logarithmically through the size of the gate set. Moreover, the sample complexity scales quadratically with the parameter $M$ and with the magnitude of the Hamiltonian coefficients $c_i$. This reflects the fact that, for each unitary $U_j$, 
only Hamiltonian terms overlapping its support contribute to the corresponding 
energy difference. As a result, the estimation error is governed by the local interaction structure of the Hamiltonian rather than by the total system size. In particular, for sparse interaction structures $M$ remains bounded as the system size grows, whereas for dense objectives $M$ may grow with system size, leading to higher measurement cost.
Finally, we emphasize that all bounds presented above are worst-case guarantees and are generally not tight. In practice, exploiting the structure of a specific Hamiltonian can substantially reduce the number of measurement shots required; see \cref{app:tfi} for an explicit example based on the transverse-field Ising model. Full proofs are deferred to \cref{appendix:proofs}. 

\section{Numerical Experiments}
\label{sec:numerics}
We now evaluate the performance of SEGQE in numerical simulations. All results are obtained from noiseless, exact state-vector simulations.
\subsection{Transverse-Field Ising Model}
\label{sec:TFI}
To benchmark SEGQE, we approximate the ground state and ground-state energy of the transverse-field Ising (TFI) model with open boundary conditions, defined by
\begin{equation}
    \label{eq:TFI}
    H=w\sum^n_{i=1}Z_i+J\sum^{n-1}_{i=1}X_iX_{i+1},
\end{equation}
where $w$ denotes the transverse-field strength, $J$ the nearest neighbor coupling constant, and $n$ the number of spins (qubits). The gate set considered by SEGQE in this experiment includes all Pauli rotations acting on most $m=2$ qubits. The per-iteration measurement cost is therefore upper-bounded by \cref{cor:pauli-rotations}. Exploiting the structure of the Hamiltonian reduces this bound by approximately a factor of $6$ (see \cref{app:tfi}). 
\begin{figure}[htbp]
  \includegraphics{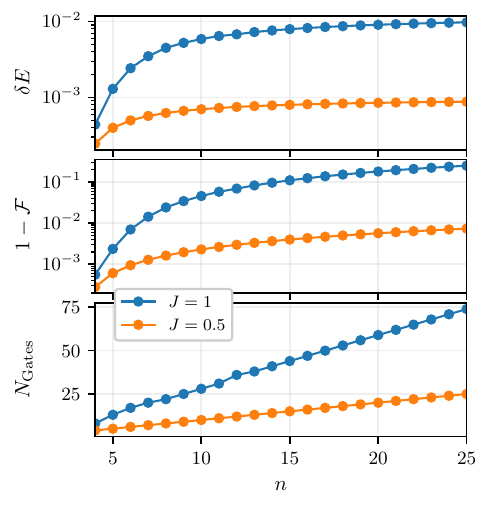}
  \vspace{-1em}
  \caption{Convergence properties of SEGQE applied to the open-boundary transverse-field Ising model at criticality ($J=w=1$, blue circles) and a gapped, field-dominated regime ($J=w/2$, orange circles). (Top) Relative energy error $\delta E=(E_\mathrm{exact}{-}E_\mathrm{SEGQE})/E_\mathrm{exact}$ at convergence as a function of system size $n$. (Middle) Infidelity with the exact ground state $1-\mathcal{F}$. (Bottom) Number of appended two-qubit Pauli rotations required for convergence $N_\mathrm{Gates}$. In all cases, SEGQE terminates once no candidate gate yields an energy reduction exceeding $\Delta=10^{-3}$.}
  \label{fig:2a}
\end{figure}
\Cref{fig:2a} shows the convergence properties of SEGQE applied to the open-boundary TFI model at the critical point $J=w=1$ and in a gapped, field-dominated regime $J=w/2$. In both settings, the number of SEGQE iterations, equivalently the number of appended gates $N_\mathrm{Gates}$, grows approximately linearly with system size $n$, with a slope at criticality roughly three times larger than in the gapped regime. This behavior is expected: away from criticality, correlations are short-ranged and few local rotations suffice to capture the ground state, whereas at criticality longer-range correlations emerge and must be incorporated through additional gates. 

Consistent with this picture, inspecting the circuits generated by SEGQE reveals that in the gapped regime the algorithm converges after appending only nearest-neighbor two-qubit Pauli rotations. At the critical point, SEGQE subsequently appends two additional sequences of gates acting on next-nearest and next-next-nearest neighbor qubit pairs, followed by a small number of corrective gates before terminating (see \cref{app:tfi}). This explains the larger prefactor in the linear scaling and demonstrates that SEGQE automatically adapts the circuit structure to the correlation properties of the target ground state. 

\Cref{fig:2a} (top) visualizes the behavior of the relative energy error $\delta E:=(E_\mathrm{exact}{-}E_\mathrm{SEGQE})/E_\mathrm{exact}$ as a function of system size. In both regimes, $\delta E$ initially increases with $n$ before leveling off, remaining below $1\%$ at criticality and below $0.1\%$ in the gapped phase across the system sizes considered. This indicates that SEGQE maintains controlled energy accuracy as the system size increases. 
Motivated by the need for approximate ground states as initial states for fault-tolerant algorithms, such as phase estimation, \cref{fig:2a} (middle) also reports the infidelity with the exact ground state, $1-\mathcal{F}=1-\left|\langle \psi_\mathrm{SEGQE}|\psi_\mathrm{exact}\rangle \right|^2$. For context, the success probability of phase estimation scales linearly with the fidelity $\mathcal{F}$, so the required repetition overhead scales as $1/\mathcal{F}$. As expected, the infidelity increases with system size, reflecting the accumulation of local imperfections. Nevertheless, this increase remains gradual over the system sizes considered and contrasts with the exponentially small overlap expected for generic random states. This suggests that SEGQE prepares structured approximations to the ground state that retain nontrivial overlap with the true ground state and may serve as suitable starting points for further processing.

\begin{figure}[htbp]
  \centering
  \includegraphics{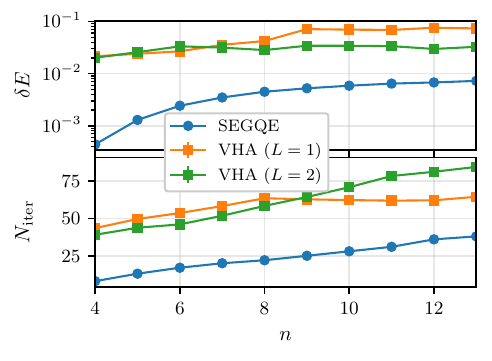}
  \vspace{-2em}
  \caption{Qualitative comparison between SEGQE and a variational quantum eigensolver (VQE) based on a variational Hamiltonian ansatz (VHA) applied to the open-boundary transverse-field model at the critical point $J=w=1$. (Top) Relative energy error $\delta E = (E_\mathrm{exact}{-}E)/E_\mathrm{exact}$ at convergence as a function of system size $n$. (Bottom) Number of iterations $N_\mathrm{iter}$ required for convergence. Results are shown for SEGQE (blue circles) and for VHA-VQE with one ($L=1$, orange squares) and two ($L=2$, green squares) ansatz layers. For VHA-VQE, circuit parameters are initialized randomly and results are averaged over 300 independent initializations. Error bars indicate the standard error of the mean but are smaller than the marker size. Both algorithms terminate once the energy difference between two successive iterations falls below $\Delta=10^{-3}$.
  }
  \label{fig:2b}
\end{figure}
\Cref{fig:2b} provides a qualitative comparison between SEGQE and a VQE based on a variational Hamiltonian ansatz (VHA) at criticality ($w=J=1$). We consider two VHA circuits with different expressivities, namely a single-layer ($L=1$) and a two-layer ($L=2$) ansatz. The VHA circuit is optimized using a standard gradient-based method (ADAM, with parameter-shift gradients), with all parameters initialized randomly, and results are averaged over 300 independent initializations. Across all considered systems, SEGQE converges in fewer iterations and achieves lower final energies than VHA-VQE, illustrating that SEGQE can produce competitive ground-state estimates using a greedy circuit construction. Although iteration counts are not directly comparable due to fundamentally different optimization strategies, the theoretical results of \cref{sec:guarantees} provide important context: the per-iteration sample complexity of SEGQE scales only logarithmically with system size, whereas gradient-based VHA-VQE typically requires at least $\mathcal{O}(n)$ measurements per iteration. 
\subsection{Random Local Hamiltonians}
While the TFI model provides a structured and physically motivated benchmark, it does not probe performance in unstructured settings. We therefore study SEGQE on an ensemble of random local Hamiltonians defined as
\begin{equation}
\label{eq:RLH}
\begin{aligned}  
    H=&\sum^n_{i=1}\sum_{\alpha\in\{x,y,z\}}w_i^\alpha \sigma^\alpha_i\\+&\sum^n_{i=1}\sum_{\alpha,\beta\in\{x,y,z\}}J^{\alpha \beta}_i\sigma_i^\alpha\sigma_{(i+1)\mod n}^\beta,
\end{aligned} 
\end{equation}
where the coefficients $w_{i}^\alpha$ and $J_i^{\alpha \beta}$ are drawn independently from $\mathcal{N}(0,1)$. Here, $\sigma^\alpha_i$ denotes the Pauli operator $\sigma^\alpha\in \{X,Y,Z\}$ acting on qubit $i$, and periodic boundary conditions are imposed.
\begin{figure}
  \centering
  \includegraphics{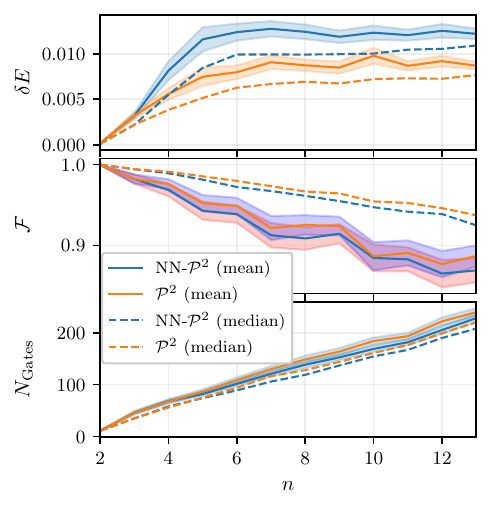}
  \vspace{-1em}
  \caption{Convergence properties of SEGQE applied to 500 instances of the random local Hamiltonians defined in \cref{eq:RLH}. (Top) Relative energy error $\delta E=(E_\mathrm{exact}{-}E_\mathrm{SEGQE})/E_\mathrm{exact}$ at convergence as a function of system size $n$. (Middle) Ground-state fidelity $\mathcal{F}$. (Bottom) Number of appended two-qubit gates required for convergence $N_\mathrm{Gates}$. Results are shown for SEGQE using all two-qubit Pauli rotations ($\mathcal{P}^2$, orange) and for the restriction to nearest-neighbor two-qubit Pauli rotations (NN-$\mathcal{P}^2$, blue). Solid lines indicate the mean over all considered instances, dashed lines the median. Shaded regions indicate the two-sigma confidence interval of the mean. In all cases, SEGQE terminates once no candidate gate yields a maximal energy reduction exceeding $\Delta=10^{-3}$.
  }
  \label{fig:a}
\end{figure}

\Cref{fig:a} shows the convergence properties of SEGQE applied to the ensemble of local Hamiltonians defined in \cref{eq:RLH}, considering two different gate sets: all two-qubit Pauli rotations and their restriction to nearest-neighbors. In both cases, the number of appended gates required for convergence grows approximately linearly with system size $n$, consistent with the behavior observed for the TFI model. Across the system sizes considered, the relative energy error depends only weakly on $n$ and appears to saturate at the percent level, while the ground-state fidelity decreases approximately linearly. Restricting SEGQE to nearest-neighbor gates leads to systematically worse energy accuracy and ground-state fidelity, indicating that longer-range two-qubit rotations are important to accurately capture the ground states of these unstructured Hamiltonians. 

The weak system-size dependence of the final energy error can be understood from the local and random nature of the considered Hamiltonians. In one-dimensional chains with random couplings, ground-state correlations typically decay rapidly with distance, yielding a finite correlation length on average. Importantly, this does not imply that correlations are strictly limited to nearest neighbors. Rather, accurately capturing the ground state generally requires incorporating correlations over a small but finite range. This is consistent with our observation that allowing a limited number of longer-range two-qubit rotations significantly improves energy accuracy, while further extensions of the interaction range yield diminishing returns, explaining the observed mild growth and apparent saturation of the energy error with increasing system size.

To further investigate the dependence of SEGQE on the choice of gate set, \cref{fig:b} plots the relative energy error $\delta E$ as a function of the number of appended gates $N_\mathrm{Gates}$ at fixed system size ($n=5$) for four distinct gate sets: 2-qubit Pauli rotations ($\mathcal{P}^2$), arbitrary 2-qubit gates ($\mathrm{SU}(4)$), and their restrictions to nearest-neighbor qubits, denoted NN-$\mathcal{P}^2$ and NN-$\mathrm{SU}(4)$. 
\begin{figure}[htbp]
  \centering
  \includegraphics{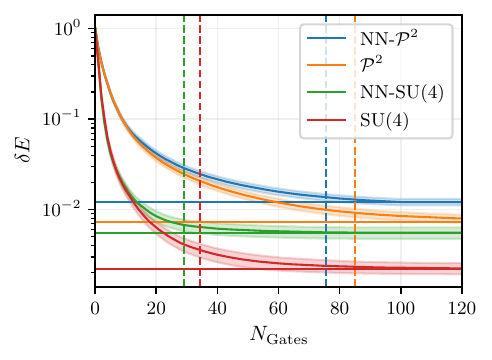}
  \caption{Dependence of SEGQE convergence on gate-set for random local Hamiltonians at fixed system size $n=5$. Shown is the relative energy error $\delta E=(E_\mathrm{exact}{-}E_\mathrm{SEGQE})/E_\mathrm{exact}$ as a function of the number of appended gates $N_\mathrm{Gates}$ for four gate sets: nearest-neighbor two-qubit rotations (NN-$\mathcal{P}^2$, blue), arbitrary-distance two-qubit Pauli rotations ($\mathcal{P}^2$, orange), nearest-neighbor arbitrary two-qubit unitaries (NN-$\mathrm{SU}(4)$, green), and arbitrary-distance two-qubit unitaries ($\mathrm{SU}(4)$, red). Solid lines indicate the mean over 100 independent random Hamiltonian instances, and shaded regions indicate the one-sigma confidence interval of the mean. Horizontal lines show the average final energy error at convergence for each gate set, while vertical dashed lines indicate the corresponding average number of appended gates at convergence. SEGQE terminates once no candidate gate yields a maximal energy reduction exceeding $\Delta=10^{-3}$. 
  }
  \label{fig:b}
\end{figure}
Two distinct effects can be observed. First, for a fixed gate geometry, allowing arbitrary two-qubit unitaries leads to lower final energies and faster convergence, reflecting the increased local expressivity when considering all of $\mathrm{SU}(4)$. Second, restricting the gate set to nearest neighbors results in earlier saturation at higher energy, indicating that geometric restrictions limit the correlations that can be captured. The plot further shows that while nearest-neighbor variants converge more rapidly, this behavior reflects premature convergence rather than more efficient optimization. Overall, these results indicate that both increased local gate expressivity and longer-range interactions improve the achievable energy accuracy, highlighting that SEGQE should exploit richer gate sets when available. Finally, we note that in the simulations shown here, the classical optimization over $\mathrm{SU}(4)$ required fewer than $T=2\times 10^4$ function evaluations for each appended gate. Although this may appear costly, the logarithmic dependence of the SEGQE sample complexity on the size of the gate set (see \cref{sec:guarantees}) ensures that this overhead remains moderate. For example, fixing the confidence parameter to $\delta=0.01$, increasing the optimization budget from $T=100$ to $T=2\times 10^4$ increases the required number of measurements by a factor of less than $1.5$ for the $5$-qubit system considered in \cref{fig:b}. Assuming that the number of required function evaluations does not increase with system size, this relative overhead becomes even less pronounced for larger systems.

\section{Conclusion}
\label{sec:conclusion}
In this work, we introduced the Shadow Enhanced Greedy Quantum Eigensolver (SEGQE), a greedy, shadow-assisted hybrid algorithm for measurement-efficient ground-state preparation. 
The central feature of SEGQE is its favorable measurement scaling. At each iteration, a single batch of classical shadows suffices to evaluate a large number of candidate gates entirely in classical post-processing. We proved rigorous worst-case bounds showing that the per-iteration measurement cost depends only logarithmically on the number of candidate gates. In addition, our numerical experiments on transverse-field Ising models and ensembles of random local Hamiltonians indicate that the required number of iterations grows approximately linearly with system size. Together, these results suggest an overall scaling behavior that is particularly well aligned with early fault-tolerant regimes, where logical measurement rates are relatively slow while moderately deep logical circuits are feasible.
Furthermore, SEGQE cleanly separates quantum data acquisition from the computationally intensive search over gate additions, enabling the efficient integration of high-performance computing resources and shifting the dominant optimization cost to classical hardware. 

The flexibility of the SEGQE framework further allows straightforward adaptation to hardware-native logical gate sets and fault-tolerant compilation constraints, making it compatible with realistic fault-tolerant architectures. Future work includes tightening instance-dependent performance guarantees, extending the framework to related tasks such as excited-state preparation, and developing parallelizable classical optimization strategies tailored to shadow-based objective estimation.
Overall, SEGQE illustrates how classical shadows can serve as a core primitive in adaptive, measurement-efficient state-preparation strategies for fault-tolerant quantum algorithms.

\section{Acknowledgments}
This research was supported by a grant funded by Moderna and Quantum Motion for an industrial PhD Studentship.
B.K. thanks UKRI for the Future Leaders Fellowship Theory to Enable Practical Quantum Advantage (MR/Y015843/1).
The authors also acknowledge funding from the
EPSRC projects Robust and Reliable Quantum Computing (RoaRQ, EP/W032635/1)
and Software Enabling Early Quantum Advantage (SEEQA, EP/Y004655/1).
This research was funded in part by UKRI (MR/Y015843/1).
For the purpose of Open Access, the author has applied a CC BY public copyright licence
to any Author Accepted Manuscript version arising from this submission.

\bibliographystyle{apsrev4-2}
\bibliography{refs}

@article{farhi2014quantum,
  title={A quantum approximate optimization algorithm},
  author={Farhi, Edward and Goldstone, Jeffrey and Gutmann, Sam},
  journal={arXiv preprint arXiv:1411.4028},
  year={2014}
}

@article{peruzzo2014variational,
  title={A variational eigenvalue solver on a photonic quantum processor},
  author={Peruzzo, Alberto and McClean, Jarrod and Shadbolt, Peter and Yung, Man-Hong and Zhou, Xiao-Qi and Love, Peter J and Aspuru-Guzik, Al{\'a}n and O’brien, Jeremy L},
  journal={Nature Communications},
  volume={5},
  number={1},
  pages={4213},
  year={2014},
  publisher={Nature Publishing Group UK London}
}

@article{cerezo2021variational,
  title={Variational quantum algorithms},
  author={Cerezo, Marco and Arrasmith, Andrew and Babbush, Ryan and Benjamin, Simon C and Endo, Suguru and Fujii, Keisuke and McClean, Jarrod R and Mitarai, Kosuke and Yuan, Xiao and Cincio, Lukasz and others},
  journal={Nature Reviews Physics},
  volume={3},
  number={9},
  pages={625--644},
  year={2021},
  publisher={Nature Publishing Group UK London}
}

@article{huang2020predicting,
  title={Predicting many properties of a quantum system from very few measurements},
  author={Huang, Hsin-Yuan and Kueng, Richard and Preskill, John},
  journal={Nature Physics},
  volume={16},
  number={10},
  pages={1050--1057},
  year={2020},
  publisher={Nature Publishing Group UK London}
}

@article{bittel2021training,
  title={Training variational quantum algorithms is NP-hard},
  author={Bittel, Lennart and Kliesch, Martin},
  journal={Physical Review Letters},
  volume={127},
  number={12},
  pages={120502},
  year={2021},
  publisher={APS}
}

@article{kandala2017hardware,
  title={Hardware-efficient variational quantum eigensolver for small molecules and quantum magnets},
  author={Kandala, Abhinav and Mezzacapo, Antonio and Temme, Kristan and Takita, Maika and Brink, Markus and Chow, Jerry M and Gambetta, Jay M},
  journal={Nature},
  volume={549},
  number={7671},
  pages={242--246},
  year={2017},
  publisher={Nature Publishing Group}
}

@article{moll2018quantum,
  title={Quantum optimization using variational algorithms on near-term quantum devices},
  author={Moll, Nikolaj and Barkoutsos, Panagiotis and Bishop, Lev S and Chow, Jerry M and Cross, Andrew and Egger, Daniel J and Filipp, Stefan and Fuhrer, Andreas and Gambetta, Jay M and Ganzhorn, Marc and others},
  journal={Quantum Science and Technology},
  volume={3},
  number={3},
  pages={030503},
  year={2018},
  publisher={IOP Publishing}
}

@article{farhi2018classification,
  title={Classification with quantum neural networks on near term processors},
  author={Farhi, Edward and Neven, Hartmut},
  journal={arXiv preprint arXiv:1802.06002},
  year={2018}
}

@article{cong2019quantum,
  title={Quantum convolutional neural networks},
  author={Cong, Iris and Choi, Soonwon and Lukin, Mikhail D},
  journal={Nature Physics},
  volume={15},
  number={12},
  pages={1273--1278},
  year={2019},
  publisher={Nature Publishing Group UK London}
}

@article{wang2022state,
  title={State preparation boosters for early fault-tolerant quantum computation},
  author={Wang, Guoming and Sim, Sukin and Johnson, Peter D},
  journal={Quantum},
  volume={6},
  pages={829},
  year={2022},
  publisher={Verein zur F{\"o}rderung des Open Access Publizierens in den Quantenwissenschaften}
}

@article{romero2018strategies,
  title={Strategies for quantum computing molecular energies using the unitary coupled cluster ansatz},
  author={Romero, Jonathan and Babbush, Ryan and McClean, Jarrod R and Hempel, Cornelius and Love, Peter J and Aspuru-Guzik, Al{\'a}n},
  journal={Quantum Science and Technology},
  volume={4},
  number={1},
  pages={014008},
  year={2018},
  publisher={IOP Publishing}
}

@article{wecker2015progress,
  title={Progress towards practical quantum variational algorithms},
  author={Wecker, Dave and Hastings, Matthew B and Troyer, Matthias},
  journal={Physical Review A},
  volume={92},
  number={4},
  pages={042303},
  year={2015},
  publisher={APS}
}

@article{grimsley2019adaptive,
  title={An adaptive variational algorithm for exact molecular simulations on a quantum computer},
  author={Grimsley, Harper R and Economou, Sophia E and Barnes, Edwin and Mayhall, Nicholas J},
  journal={Nature Communications},
  volume={10},
  number={1},
  pages={3007},
  year={2019},
  publisher={Nature Publishing Group UK London}
}

@article{zimboras2025myths,
  title={Myths around quantum computation before full fault tolerance: What no-go theorems rule out and what they don't},
  author={Zimbor{\'a}s, Zolt{\'a}n and Koczor, B{\'a}lint and Holmes, Zo{\"e} and Borrelli, Elsi-Mari and Gily{\'e}n, Andr{\'a}s and Huang, Hsin-Yuan and Cai, Zhenyu and Ac{\'\i}n, Antonio and Aolita, Leandro and Banchi, Leonardo and others},
  journal={arXiv preprint arXiv:2501.05694},
  year={2025}
}

@article{mcclean2018barren,
  title={Barren plateaus in quantum neural network training landscapes},
  author={McClean, Jarrod R and Boixo, Sergio and Smelyanskiy, Vadim N and Babbush, Ryan and Neven, Hartmut},
  journal={Nature Communications},
  volume={9},
  number={1},
  pages={4812},
  year={2018},
  publisher={Nature Publishing Group UK London}
}

@article{larocca2024review,
  title={A review of barren plateaus in variational quantum computing},
  author={Larocca, Martin and Thanasilp, Supanut and Wang, Samson and Sharma, Kunal and Biamonte, Jacob and Coles, Patrick J and Cincio, Lukasz and McClean, Jarrod R and Holmes, Zo{\"e} and Cerezo, M},
  journal={arXiv preprint arXiv:2405.00781},
  year={2024}
}

@article{cerezo2021cost,
  title={Cost function dependent barren plateaus in shallow parametrized quantum circuits},
  author={Cerezo, Marco and Sone, Akira and Volkoff, Tyler and Cincio, Lukasz and Coles, Patrick J},
  journal={Nature Communications},
  volume={12},
  number={1},
  pages={1791},
  year={2021},
  publisher={Nature Publishing Group UK London}
}

@article{holmes2022connecting,
  title={Connecting ansatz expressibility to gradient magnitudes and barren plateaus},
  author={Holmes, Zo{\"e} and Sharma, Kunal and Cerezo, Marco and Coles, Patrick J},
  journal={PRX Quantum},
  volume={3},
  number={1},
  pages={010313},
  year={2022},
  publisher={APS}
}

@article{wang2021noise,
  title={Noise-induced barren plateaus in variational quantum algorithms},
  author={Wang, Samson and Fontana, Enrico and Cerezo, Marco and Sharma, Kunal and Sone, Akira and Cincio, Lukasz and Coles, Patrick J},
  journal={Nature Communications},
  volume={12},
  number={1},
  pages={6961},
  year={2021},
  publisher={Nature Publishing Group UK London}
}

@article{anschuetz2022quantum,
	title={Quantum variational algorithms are swamped with traps},
	author={Anschuetz, Eric R and Kiani, Bobak T},
	journal={Nature Communications},
	volume={13},
	number={1},
	pages={7760},
	year={2022},
	publisher={Nature Publishing Group UK London}
}

@article{hwang2025preparing,
	title={Preparing ground and excited states using adiabatic CoVaR},
	author={Hwang, Wooseop and Koczor, B{\'a}lint},
	journal={New Journal of Physics},
	volume={27},
	number={2},
	pages={023025},
	year={2025},
	publisher={IOP Publishing}
}

@article{koczor2022quantum,
	title={Quantum analytic descent},
	author={Koczor, B{\'a}lint and Benjamin, Simon C},
	journal={Physical Review Research},
	volume={4},
	number={2},
	pages={023017},
	year={2022},
	publisher={APS}
}

@article{koczor2022quantum2,
	title={Quantum natural gradient generalized to noisy and nonunitary circuits},
	author={Koczor, B{\'a}lint and Benjamin, Simon C},
	journal={Physical Review A},
	volume={106},
	number={6},
	pages={062416},
	year={2022},
	publisher={APS}
}

@article{boyd2022training,
  title={Training variational quantum circuits with CoVaR: Covariance root finding with classical shadows},
  author={Boyd, Gregory and Koczor, B{\'a}lint},
  journal={Physical Review X},
  volume={12},
  number={4},
  pages={041022},
  year={2022},
  publisher={APS}
}

@article{feniou2025greedy,
  title={Greedy gradient-free adaptive variational quantum algorithms on a noisy intermediate scale quantum computer},
  author={Feniou, C{\'e}sar and Hassan, Muhammad and Claudon, Baptiste and Courtat, Axel and Adjoua, Olivier and Maday, Yvon and Piquemal, Jean-Philip},
  journal={Scientific Reports},
  volume={15},
  number={1},
  pages={18689},
  year={2025},
  publisher={Nature Publishing Group UK London}
}

@article{nakanishi2020sequential,
  title={Sequential minimal optimization for quantum-classical hybrid algorithms},
  author={Nakanishi, Ken M and Fujii, Keisuke and Todo, Synge},
  journal={Physical Review Research},
  volume={2},
  number={4},
  pages={043158},
  year={2020},
  publisher={APS}
}

@article{motta2020determining,
  title={Determining eigenstates and thermal states on a quantum computer using quantum imaginary time evolution},
  author={Motta, Mario and Sun, Chong and Tan, Adrian TK and O’Rourke, Matthew J and Ye, Erika and Minnich, Austin J and Brandao, Fernando GSL and Chan, Garnet Kin-Lic},
  journal={Nature Physics},
  volume={16},
  number={2},
  pages={205--210},
  year={2020},
  publisher={Nature Publishing Group UK London}
}

@article{puig2026warm,
  title={Warm Starts, Cold States: Exploiting Adiabaticity for Variational Ground-States},
  author={Puig, Ricard and Casas, Berta and Cervera-Lierta, Alba and Holmes, Zo{\"e} and P{\'e}rez-Salinas, Adri{\'a}n},
  journal={arXiv preprint arXiv:2602.06137},
  year={2026}
}

\onecolumngrid
\appendix
\crefalias{section}{appendix}

\section{Related Work \label{appendix:related_works}}
Iterative state-preparation methods such as quantum imaginary time evolution (QITE) \cite{motta2020determining} also construct circuits step by step, but they aim to approximate a global imaginary-time flow. To this end, QITE approximates local non-unitary operations with unitary circuits obtained from increasingly non-local optimizations, making the direct integration of measurement-efficient frameworks such as classical shadows challenging. In contrast, SEGQE performs greedy local optimization based on instantaneous energy reduction and maintains locality throughout the evolution. 

Similarly, the greedy, gradient-free, adaptive variational quantum eigensolver proposed in Ref.~\cite{feniou2025greedy} appends gates, specifically those admitting analytical optimization such as Pauli rotations, based on instantaneous energy reduction, but relies on conventional measurement strategies. SEGQE extends this framework by incorporating classical shadows, providing rigorous per-iteration sample-complexity guarantees, and accommodating large, collections of expressive candidate gates.

\section{Proofs}
\label{appendix:proofs}
Throughout this section, as well as in the results section of the main text, the term classical shadow refers to an estimator obtained from a single uniformly random Pauli measurement followed by the corresponding inverse measurement channel \cite{huang2020predicting}. For any observable $O$, we denote by $o := \Tr(\rho O)$ its expectation value in a given quantum state $\rho$ (or $\ket{\psi}$), and by $\hat{o} := \Tr(O\hat{\rho})$ the corresponding single-shot estimator obtained from a classical shadow $\hat{\rho}$. Unless stated otherwise, expectation values $\mathbb{E}[\cdot]$ are taken with respect to the distribution of such single-shot Pauli shadows of a fixed underlying state $\rho$. Furthermore, we denote by $\mathcal{P}^{n}$ the $n$-qubit Pauli group and by $\mathrm{Cl}(2^n)$ the $n$-qubit Clifford group. 

\begin{lemma}[Second Moments of Pauli Shadows]
    \label{lemma:second_moments}
    Let $\hat{\rho}$ be a classical shadow of an $n$-qubit quantum state $\rho$ obtained from a single uniformly random local Pauli measurement. For $P_\alpha, P_\beta \in \mathcal{P}^n$, define $\hat{p}_\alpha=\mathrm{Tr}(P_\alpha\hat{\rho})$ and $\hat{p}_\beta=\mathrm{Tr}(P_\beta\hat{\rho})$. Then
    \begin{equation}
        \mathbb{E}\left[\hat{p}_\alpha \hat{p}_\beta
        \right]=\begin{cases}
        3^L\mathrm{Tr}(\rho P_\alpha P_\beta) & \text{if $P_\alpha$ and $P_\beta$ commute on every qubit}\\
        0 &\text{otherwise}
        \end{cases},
    \end{equation}
    where $L$ is the number of qubits on which both $P_\alpha$ and $P_\beta$ act nontrivially.
\end{lemma}
\begin{proof}
    Let $w_\alpha$, respectively $w_\beta$, denote the localities of the two Pauli operators $P_\alpha,P_\beta\in\mathcal{P}^n$, and let $P_\alpha = \bigotimes^n_{j=1}P_\alpha^j$ and $P_\beta = \bigotimes^n_{j=1}P_\beta^j$, where $P^j_\alpha,P^j_\beta\in\mathcal{P}$. Furthermore, let $\ket{b}\in\{\ket{0},\ket{1}\}^{\otimes n}$ denote the measurement outcome and $U\in\{I,H,S^\dagger H\}^{\otimes n}\subset \mathrm{Cl}(2)^{\otimes n}$ the measurement basis associated with the classical shadow $\hat{\rho}$. Moreover we define $\mathcal{B}=\{\ket{0},\ket{1},\ket{\pm}, \ket{\pm i}\}$ as the set of all six single-qubit Pauli eigenstates. In \cite{huang2020predicting} it was shown that for any $k$-local Pauli operator $P$, we can write $\hat{p}:=\mathrm{Tr}(P\hat{\rho})=3^k\bra{b}U P U^\dagger\ket{b}$. Taking the expectation value of the product $\hat{p}_\alpha \hat{p}_\beta$ with respect to the measurement basis and the measurement outcome then yields
    \begin{align}
        \mathbb{E}\left[\hat{p}_\alpha \hat{p}_\beta\right]&=\frac{3^{w_\alpha + w_\beta}}{3^n}\sum_{\ket{x}\in \mathcal{B}^{\otimes n}}\bra{x}\rho\ket{x}\bra{x}P_\alpha\ket{x}\bra{x}P_\beta\ket{x} \\
        &=\frac{3^{w_\alpha + w_\beta}}{3^n}\mathrm{Tr}\left[\rho\bigotimes^n_{j=1}\left(\sum_{\ket{x^j}\in \mathcal{B}}\ket{x^j}\bra{x^j}\bra{x^j}P_\alpha^j\ket{x^j}\bra{x^j}P^j_\beta\ket{x^j} \right) \right]  \\
        &=\frac{3^{w_\alpha + w_\beta}}{3^n}\mathrm{Tr}\left[\rho\bigotimes^n_{j=1}\left(\begin{cases}
            I & \text{if } P^j_\alpha=P^j_\beta\neq I \\
        3I & \text{if } P^j_\alpha=P^j_\beta= I \\
        P^j_\alpha & \text{if } P^j_\alpha\neq I\text{ and }P^j_\beta= I \\
        P^j_\beta & \text{if } P^j_\beta\neq I\text{ and }P^j_\alpha= I \\
        0 & \text{otherwise }\Leftrightarrow [P^j_\alpha,P^j_\beta]\neq0
        \end{cases} \right) \right]\\
        &=\begin{cases}
            3^L\mathrm{Tr}\left(\rho P_\alpha P_\beta\right)& \text{if }[P^j_\alpha, P^j_\beta]=0\quad\forall1\leq j\leq n\\
        0 &\text{otherwise}
        \end{cases},
    \end{align}
    where $L$ is the number of qubits on which both $P_\alpha$ and $P_\beta$ act nontrivially.
\end{proof}

Since classical shadows define unbiased estimators, the covariance between two Pauli estimators follows directly from \cref{lemma:second_moments} and is stated in \cref{cor:covariances}.

\begin{corollary}[Covariances of Pauli Shadows]
    \label{cor:covariances}
    Under the assumptions of \cref{lemma:second_moments}, it follows
    \begin{equation}
        \mathrm{Cov}\left(\hat{p}_\alpha, \hat{p}_\beta
        \right)=\begin{cases}
        3^L\mathrm{Tr}(\rho P_\alpha P_\beta)-\mathrm{Tr}(\rho P_\alpha)\mathrm{Tr}(\rho P_\beta) & \text{if $P_\alpha$ and $P_\beta$ commute on every qubit}\\
        -\mathrm{Tr}(\rho P_\alpha)\mathrm{Tr}(\rho P_\beta) &\text{otherwise}
        \end{cases}.
    \end{equation}
\end{corollary}

For the special case $P_\alpha = P_\beta$, \cref{cor:covariances} reduces to the known single-shot variance of Pauli shadow estimators derived in \cite{huang2020predicting}. 

These results can now be used to derive upper bounds on the variance of linear functions of Pauli estimators. To define the single-shot estimator for the energy differences considered by SEGQE, we briefly recall how SEGQE constructs energy-difference estimators (see \cref{sec:procedure} for full details). Given an $n$-qubit Hamiltonian $H=\sum^r_{i=1}c_iP_i$ and an $n$-qubit unitary of locality $m$, only Hamiltonian terms whose support overlaps nontrivially with that of $U$ contribute to the energy difference. For each such term, the conjugated operator $U^\dagger P_i U$ is expanded in the Pauli basis on the support of $U$, yielding a linear representation of the energy difference in terms of Pauli expectation values. Replacing each Pauli expectation value by its corresponding single-shot Pauli shadow estimator defines a linear single-shot estimator $\hat{\Delta E}$ for the energy difference. Using this construction, the following lemma provides an upper bound on the variance of these single-shot energy difference estimators.

\begin{lemma}[Variance of Energy-Difference Estimators]
    \label{lemma:variance_energy}
    Consider an $n$-qubit Hamiltonian $H=\sum^r_{i=1}c_iP_i$, where each $P_i$ is a Pauli operator of locality at most $l$, an $n$-qubit quantum state $\ket{\psi}$, and an arbitrary $n$-qubit gate $U$ of locality $m$. Let $M^u$ denote the number of Hamiltonian terms whose support overlaps nontrivially with the support of $U$, and define $c_\mathrm{max}:=\max_{1\leq i\leq r}|c_i|$. Then the variance of the single-shot estimator $\hat{\Delta E}$ of the energy difference
    \begin{equation}
        \Delta E= \bra{\psi}H\ket{\psi}-\bra{\psi}U^\dagger H U \ket{\psi},
    \end{equation}
    with respect to a single classical shadow, is bounded as
    \begin{equation}
        \mathrm{Var}\left[\hat{\Delta E}\right]\leq 4^{m+1}3^{l-1}\left(M^uc_\mathrm{max}\right)^2.
    \end{equation}
\end{lemma}

\begin{proof}
    Let $Q^i\subseteq\{1,\ldots n\}$ denote the support of the Hamiltonian term $P_i$ and let $Q^u\subseteq\{1,\ldots n\}$ denote the support of the unitary $U$. $Q^c$ denotes the complement of $Q^u$. Only Hamiltonian terms with $Q^i\cap Q^u\neq\emptyset$ contribute to the energy difference, and we denote their indices by $I^u:=\{i:Q^i\cap Q^u\neq\emptyset\}$. Writing $P_i=P_i^u\otimes P_i^c$ with respect to the bipartition $Q^u\cup Q^c$, the energy difference induced by $U$ can be expressed as
    \begin{align}
        \Delta E&= \sum_{i\in I^u}c_i\left(\bra{\psi}P^u_i\otimes P_i^c\ket{\psi}-\bra{\psi}U^\dagger P^u_i U\otimes P_i^c\ket{\psi} \right).
    \end{align}
    Expanding the conjugated operator $U^\dagger P_i^u U$ in the Pauli basis on the $m$ qubits in $Q^u$ yields
    \begin{equation}
        U^\dagger P^u_i U=\sum_{j=0}^{4^m-1}r_{ij}P_j^u,\qquad r_{ij}:=\frac{1}{2^m}\mathrm{Tr}\left[U^\dagger P^u_iUP_j^u\right],
    \end{equation}
    and therefore
    \begin{equation}
        \label{eq:energy_difference_app}
        \Delta E=\sum_{i\in I^u}\sum^{4^m-1}_{j=0}f_{ij}p_{ij},
    \end{equation}
    where $f_{ij}:= c_i(\delta_{P^u_iP_j^u} - r_{ij})$ and $p_{ij} := \bra{\psi} P_j^u \otimes P_i^c \ket{\psi}$. This is exactly the linear representation of the energy differences used by SEGQE (see \cref{sec:procedure} for details). For notational convenience, we fold the double index into $\alpha=(i,j)$ and denote by $\mathcal{F}^u:=I^u\times \{0,...,4^m-1\}$ the set of all valid indices. This allows us to write $P_\alpha=P^u_j\otimes P^c_i$ and $\Delta E=\sum_{\alpha\in\mathcal{F}^u} f_\alpha p_\alpha$.
    
    Now, let $\hat{p}_{\alpha} := \mathrm{Tr}(\hat\rho P_\alpha)$ denote the estimator of $p_\alpha$ obtained from a single classical shadow $\hat{\rho}$ of $\ket{\psi}\bra{\psi}$. These single-shot Pauli estimators can be used to construct an estimator $\hat{\Delta E}$ of $\Delta E$ via \cref{eq:energy_difference_app}. The variance of this estimator is bounded by its second moment:
    \begin{equation}
        \mathrm{Var}[\hat{\Delta E}]=\mathbb{E}\left[\hat{\Delta E}^2\right]-\mathbb{E}\left[\hat{\Delta E}\right]^2\leq \mathbb{E}\left[\hat{\Delta E}^2\right]=\sum_{\alpha, \beta\in\mathcal{F}^u}f_\alpha\mathbb{E}\left[\hat{p}_\alpha \hat{p}_\beta\right]f_\beta.
    \end{equation}
    Using \cref{lemma:second_moments} and the fact that $\lvert\mathrm{Tr}(\rho P_\alpha P_\beta)\rvert\leq 1$ for all Pauli operators $P_\alpha$ and $P_\beta$, it follows that
    \begin{equation}
        \label{eq:app:2}
        \mathrm{Var}[\hat{\Delta E}]\leq \sum_{\alpha, \beta\in\mathcal{F}^u}f_\alpha K_{\alpha \beta}f_\beta,\qquad K_{\alpha \beta}=:\begin{cases}
        3^L&\text{if $P_\alpha$ and $P_\beta$ commute qubit-wise}\\
        0&\text{otherwise}
    \end{cases},
    \end{equation}
    where $L$ denotes the number of qubits on which both $P_\alpha$ and $P_\beta$ act nontrivially.
    \Cref{eq:app:2} can be written in matrix form as the quadratic expression $f^\top Kf$, where both, the vector $f$ and the correlation matrix $K$ are restricted to indices $\alpha\in \mathcal{F}^u$. Since $K$ is symmetric, we have
    \begin{equation}
        f^\top Kf\leq \lambda_\mathrm{max}(K) \norm{f}_2^2,
    \end{equation}
    where $\lambda_\mathrm{max}(K)$ denotes the largest eigenvalue of $K$ and $\norm{f}^2_2$ the squared $\ell_2$-norm of $f$, which is given by
\begin{align}
    \norm{f}^2_2 &= \sum_{\alpha\in \mathcal{F}^u} f_\alpha^2=\sum_{i\in I^u}c_i^2\sum^{4^m-1}_{j=0}\left(\delta_{P^u_iP^u_j}-r_{ij} \right)^2\\
    &=\sum_{i\in I^u}c_i^2\left(1-2r_{ii}+\sum^{4^m-1}_{j=0}r_{ij}^2\right).
\end{align}
Since $U^\dagger P_i^u U$ is Hermitian and unitary, and the set $\{P_j^u\}^{4^m-1}_{j=0}$ forms an orthonormal basis with respect to the Hilbert-Schmidt inner product, the coefficient vector $r_i:=(r_{ij})_j$ satisfies $\norm{r_i}_2=1$. It therefore follows that
\begin{equation}
    \norm{f}^2_2=2\sum_{i\in I^u}c_i^2\left(1-r_{ii}\right)\leq 4\sum_{i\in I^u}c_i^2.
\end{equation}
To determine $\lambda_{\mathrm{max}}(K)$ of $K$, we exploit the product structure of $\mathcal{F}^u=I^u\times \{0,...,4^m-1\}$. Let $\alpha=(i,j)$ and $\beta=(k,l)$ be two composite indices. Since $Q^u$ and $Q^c$ are disjoint, the entries of $K$ factorize as $K_{\alpha\beta}= K^u_{jl}K^c_{ik}$, where $K^u_{jl}$ accounts for the correlations on $Q^u$ and $K^c_{ik}$ for $Q^c$. Consequently, the correlation matrix decomposes into the tensor product $K=K^c\otimes K^u$. This follows because Pauli operators factorize across disjoint subsystems and classical shadow correlations respect tensor-product structure. The matrix $K^u$ describes the correlations between operators in the set $\{P^u_j\}_{j=0}^{4^m-1}$ on $Q^u$. Since this set corresponds to the full $m$-qubit Pauli basis, $K^u$ takes the form $K^u=\tilde{K}^{\otimes m}$, where 
\begin{equation}
\tilde{K}
=\begin{pmatrix}
    1&1&1&1\\1&3&0&0\\1&0&3&0\\1&0&0&3
\end{pmatrix}
\end{equation}
is the single-qubit correlation matrix in the basis $\{I,X,Y,Z\}$. The eigenvalues of $\tilde{K}$ are $\{4,3,3,0\}$, which implies $\lambda_\mathrm{max}(K)=\lambda^m_\mathrm{max}(\tilde{K})=4^m$. The matrix $K^c$ is defined by the correlations between the Pauli operators $P^c_i$ on $Q^c$. Since every operator $P_i$ for $i\in I^u$ acts non-trivially on at least one qubit in $Q^u$, the weight of $P_i^c$ is bounded by $l-1$. The entries $K^c$ are therefore bounded by $K^c_{ik}\leq 3^{l-1}$. Consequently, the largest eigenvalue $\lambda_\mathrm{max}(K^c)$ is upper-bounded by the case where all Hamiltonian terms are identical on $Q^c$ (i.e., $P^c_i=P^c_k$ for all $i,k\in I^u$), yielding an $M^u\times M^u$ matrix of identical entries $3^{l-1}$, where $M^u:=|I^u|$. It follows that $\lambda_\mathrm{max}(K^c)\leq M^u3^{l-1}$. Combining these results, we obtain $\lambda_\mathrm{max}(K)\leq M^u 4^m 3^{l-1}$ and thus the following final bound on the variance:
\begin{equation}
    \mathrm{Var}\left[\hat{\Delta E}\right]\leq  M^u4^{m+1}3^{l-1}\sum_{i\in I^u}c^2_i\leq 4^{m+1}3^{l-1}(M^uc_\mathrm{max})^2,
\end{equation}
where $c_\mathrm{max}:=\max_{1\leq i\leq r}|c_i|$.
\end{proof}
In addition to the variance bound provided by \cref{lemma:variance_energy}, proving \cref{theorem1} and \cref{theorem2} requires a bound on the magnitude of the single-shot estimator $\hat{\Delta E}$. This is given by the following lemma.
\begin{lemma}[Single-Shot Bound for Energy-Difference Estimators]
    \label{lemma:singleshot_energy}
    Under the assumptions of Lemma~\ref{lemma:variance_energy}, the single-shot Pauli shadow estimator $\hat{\Delta E}$ satisfies the almost-sure bound
    \begin{equation}
        \left|\hat{\Delta E}\right|\leq2 M^u c_{\max} \, 4^{m} 3^{\,l-1}.
    \end{equation}
\end{lemma}
\begin{proof}
    We use the same linear representation of the energy-difference estimator as in the proof of Lemma~\ref{lemma:variance_energy},
    \begin{equation}
        \hat{\Delta E}=\sum_{i\in I^u}\sum_{j=0}^{4^m-1}f_{ij}\hat{p}_{ij},
    \end{equation}
    where $f_{ij}=c_i(\delta_{P_i^u,P_j^u}-r_{ij})$ and $\hat{p}_{ij}=\mathrm{Tr}\left(\hat{\rho}P^u_j\otimes P^c_i \right)$. By the triangle inequality,
    \begin{equation}
        \label{eq:ss_triangle}
        \left|\hat{\Delta E}\right|\leq\sum_{i\in I^u}\sum_{j=0}^{4^m-1} |f_{ij}|\,|\hat p_{ij}|\leq f_\mathrm{max}\sum_{i\in I^u}\sum_{j=0}^{4^m-1}|\hat p_{ij}|,
    \end{equation}
    where
    \begin{equation}
        f_\mathrm{max}:=\max_{i\in I^u}\max_{0\leq j \leq 4^m-1}\left|f_{ij}\right|\leq 2c_\mathrm{max}.
    \end{equation}
    Let $\hat{\rho}$ be a classical shadow obtained from a single uniformly random local Pauli measurement, and let $P_S=P_S^u\otimes P^c_S$ denote the corresponding measurement Pauli. Then, for every single-shot estimator $\hat{p}_{ij}=\mathrm{Tr}\left(\hat{\rho}P^u_j\otimes P^c_i \right)$ it holds
    \begin{equation}
        \left| \hat{p}_{ij}\right|=\begin{cases}
            3^{\mathrm{wt}(P^u_j)+\mathrm{wt}(P^c_i)}&\text{if $P^u_j\otimes P^c_i$ is compatible with $P_S$}\\
            0&\text{otherwise}
        \end{cases}\leq3^{l-1}\begin{cases}
            3^{\mathrm{wt}(P^u_j)}&\text{if $P^u_j$ is compatible with $P_S^u$}\\
            0&\text{otherwise}
        \end{cases},
    \end{equation}
    where $\mathrm{wt}(P)$ denotes the locality of a given Pauli $P$. Furthermore, for every measurement basis $P_S=P^u_S\otimes P^c_S\in \mathcal{P}^n$, there are exactly $\binom{m}{k}$ $m$-qubit Pauli operators of locality $k$ which are compatible with $P^u_S$. Therefore, we find
    \begin{equation}
        \sum_{i\in I^u}\sum_{j=0}^{4^m-1}|\hat p_{ij}|\leq M^u3^{l-1}\sum^m_{k=0}\binom{m}{k}3^k=M^u3^{l-1}4^m
    \end{equation}
    Putting everything together, we obtain the final bound
    \begin{equation}
        \left|\hat{\Delta E}\right|\leq2 M^u c_{\max} \, 4^{m} 3^{\,l-1}.
    \end{equation}
\end{proof}
We now combine the variance and single-shot bounds to prove \cref{theorem1} and \cref{theorem2}.

\begin{proof}[Proof of \Cref{theorem1}]
    Consider a collection of $N$ independent classical shadows $\{\hat{\rho}^i\}_{i=1}^N$ of $\ket{\psi}$ and define
    \begin{align}
        \hat{\Delta E}^i_j:=\mathrm{Tr}\left(\hat{\rho}^iH\right)-\mathrm{Tr}\left(\hat{\rho}^iU_j^\dagger HU_j\right)=\sum_{\alpha\in \mathcal{F}^j}f_{j, \alpha}\hat{p}^i_\alpha,\qquad X^i_j:=\hat{\Delta E}^i_j-\Delta E_j,
    \end{align}
    where $\hat{p}^i_\alpha=\mathrm{Tr}\left(\hat{\rho}^iP_\alpha\right)$ and $f_{j,\alpha}$, $P_\alpha$, and $\mathcal{F}^j$ are defined as in the proof of \cref{lemma:variance_energy}. Then $X^1_j,\ldots X^N_j$ are independent and satisfy $\mathbb{E}\left[X^i_j \right]=0$. We estimate $\Delta E_j$ using the empirical mean 
    \begin{equation}
        \bar{\Delta E}_j:=\frac{1}{N}\sum_{i=1}^N \hat{\Delta E}_j^i.
    \end{equation}
    Bernstein's inequality yields, for any $\epsilon>0$,
    \begin{equation}
        \mathrm{Pr}\left(\left|\Delta E_j-\bar{\Delta E}_j \right|\geq \epsilon\right)=\mathrm{Pr}\left(\left|\frac{1}{N}\sum_{i=1}^NX_j^i \right|\geq \epsilon \right)\leq 2\exp\left(-\frac{N\epsilon^2}{2\left( V+\frac{B\epsilon }{3}\right)} \right),
    \end{equation}
    provided that $\left|X^i_j\right|\leq B$ almost surely and $\mathbb{E}\left[\left(X^i_j \right)^2\right]=\mathrm{Var}\left[\hat{\Delta E}^i_j\right]\leq V$, for all $1\leq i\leq N$ and $1\leq j \leq K$. 
    By $\cref{lemma:singleshot_energy}$ $\left| \hat{\Delta E}^i_j\right|\leq 2M_jc_\mathrm{max}4^{m_j}3^{l-1}:=B_j/2$, where $m_j$ denotes the locality of the unitary $U_j$. Furthermore, by realizing that $\Delta E_j=\mathbb{E}\left[\hat{\Delta E}^i_j\right]$, it also follows that $\left|\Delta E_j\right|\leq B_j/2$. Therefore, it follows that
    \begin{equation}
        \left| X^i_j\right|\leq B_j\leq  \max_{1\leq j\leq K}B_j=Mc_\mathrm{max}4^{m+1}3^{l-1}=:B.
    \end{equation}
    Moreover, by \cref{lemma:variance_energy},
    \begin{equation}
        \mathrm{Var}\left[\hat{\Delta E}^i_j\right]\leq 4^{m+1}3^{l-1}(Mc_\mathrm{max})^2=:V,
    \end{equation}
    and in particular $V\geq B$. 

    Using a union bound over all $K$ energy differences and restricting to $\epsilon\in (0,1]$, we obtain
    \begin{equation}
        \mathrm{Pr}\left(\max_{1\leq j\leq K}\left|\Delta E_j-\bar{\Delta E}_j \right|\geq \epsilon\right)\leq 2K\exp\left(-\frac{3N\epsilon^2}{8V} \right).
    \end{equation}
    Then for any $\delta \in (0,1]$ choosing $N$ such that 
    \begin{equation}
        \delta\leq2K\exp\left(-\frac{3N\epsilon^2}{8V} \right)\Leftrightarrow N\geq \frac{8\log(2K/\delta)}{3\epsilon^2}V=\frac{32\log(2K/\delta)}{9\epsilon^2}4^m3^l(Mc_\mathrm{max})^2,
    \end{equation}
    is sufficient to ensure
    \begin{equation}
        \mathrm{Pr}\left(\max_{1\leq j\leq K}\left|\Delta E_j-\bar{\Delta E}_j \right|\geq \epsilon\right)\leq\delta.
    \end{equation}
\end{proof}
\begin{proof}[Proof of \Cref{theorem2}]
    We can decompose the parameter-dependent energy difference $\Delta E_j(\theta)$ into the trigonometric form
    \begin{equation}
        \Delta E_j(\theta)=\frac{1}{2}\left(A_j-A_j\cos(\theta)-B_j\sin(\theta)\right), 
    \end{equation}
    where $A_j:=\bra{\psi}H\ket{\psi}-\bra{\psi}X_jHX_j\ket{\psi}$ and $B_j:=\bra{\psi}i[X_j,H]\ket{\psi}$. This allows us to determine the global maximizer $\theta^*_j$ and the corresponding global maximum $\Delta E_{j,\mathrm{max}}$ analytically as
    \begin{equation}
        \theta^*_j =\arctan2(B_j, A_j),\qquad\Delta E_{j, \mathrm{max}}=\frac{1}{2}\left( A_j+\sqrt{A_j^2+B_j^2}\right).
    \end{equation}
    Since $X_j$ is Hermitian ($X_j=X_j^\dagger$) and satisfies $X_j^2=I$, it follows that $X_j$ is unitary and has locality $m_j\leq m$. Therefore, $A_j$ is of the same form as the energy differences considered in \cref{lemma:variance_energy} and \cref{lemma:singleshot_energy}. Hence, using $m_j\leq m$ and $M_j\leq M$, we obtain the uniform bounds
    \begin{equation}
        \mathrm{Var}[\hat{A}_j]\leq 4^{m+1}3^{l-1}(Mc_\mathrm{max})^2,\qquad \left| \hat{A}_j \right|\leq 2Mc_\mathrm{max}4^m3^{l-1},
    \end{equation}
    where $\hat{A}_j$ denotes the single-shot shadow estimator of $A_j$. $B_j=\bra{\psi}i[X_j,H]\ket{\psi}$ is of slightly different form. However, since $X_j$ acts nontrivially only on $m_j$ qubits, the commutator $i[X_j, H]$ can be expanded in the Pauli basis on the support of $X_j$ yielding a linear representation $B_j=\sum_{\alpha \in \mathcal{F}^j}f'_{j,\alpha}p_\alpha$, where $\mathcal{F}^j$ and $p_\alpha$ are defined as in the proof of \cref{lemma:variance_energy}. Moreover, the corresponding coefficient vectors satisfy $\norm{f'_j}^2_2\leq 4\sum_{i\in I^j}c_i^2$ and $\max_{\alpha\in \mathcal{F}^j}\left|f'_{j,\alpha}\right|\leq 2c_\mathrm{max}$. Here, $I^j$ denotes the set of Hamiltonian indices $i$ such that the support of $P_i$ overlaps with the support of $X_j$. Therefore, by the same arguments as in \cref{lemma:variance_energy} and $\cref{lemma:singleshot_energy}$, and using $m_j\leq m$ and $M_j\leq M$, we obtain the bounds 
    \begin{equation}
        \mathrm{Var}[\hat{B}_j]\leq 4^{m+1}3^{l-1}(Mc_\mathrm{max})^2,\qquad \left| \hat{B}_j \right|\leq 2Mc_\mathrm{max}4^m3^{l-1},
    \end{equation}
    where $\hat{B}_j$ denotes the single-shot shadow estimator of $B_j$. Therefore, similarly to the proof of \cref{theorem1}, we can now apply Bernstein's inequality and a union bound over all $2K$ operators to conclude that for any $\epsilon_F, \delta \in (0,1]$
    \begin{equation}
        N=\frac{32\log(4K/\delta)}{9\epsilon^2_F}4^m3^lM^2c^2_\mathrm{max}
    \end{equation}
    independent classical shadows are sufficient to ensure
    \begin{equation}
        \label{eq:event}
        \mathrm{Pr}\left(\max_{1\leq j\leq K}\max\left(\left|A_j-\bar{A}_j\right|, \left|B_j-\bar{B}_j\right|\right)\geq \epsilon_F\right)\leq \delta,
    \end{equation}
    where $\bar{A}_j$ and $\bar{B}_j$ denote the empirical means over all $N$ single-shot estimators. 
    We can use $\bar{A}_j$ and $\bar{B}_j$ to define the parameter-dependent energy difference estimator
    \begin{equation}
        \bar{\Delta E}_j(\theta)=\frac{1}{2}(\bar{A}_j-\bar{A}_j\cos(\theta)-\sin(\theta)\bar{B}_j), 
    \end{equation}
    whose global maximizer $\bar{\theta}^*_j$ and corresponding global maximum are given by
    \begin{equation}
        \bar{\theta}^*_j=\arctan2(\bar{B}_j, \bar{A}_j),\qquad \bar{\Delta E}_{j, \mathrm{max}}=\frac{1}{2}\left(\bar{A}_j+\sqrt{\bar{A}^2_j+\bar{B}^2_j} \right).
    \end{equation}
    Therefore, using the inequality $\left|\norm{x}_2-\norm{y}_2\right|\leq \norm{x-y}_2$ for the Euclidean norm, we find
    \begin{equation}
        |\Delta E_{j, \mathrm{max}}-\bar{\Delta E}_{j, \mathrm{max}}|=\frac{1}{2}\left|A_j-\bar{A_j}+\sqrt{A^2_j+B^2_j}-\sqrt{\bar{A}^2_j+\bar{B}^2_j} \right|\leq \frac{\left|A_j-\bar{A}_j\right|+\sqrt{\left(A_j-\bar{A}_j\right)^2+\left(B_j-\bar{B}_j\right)^2}}{2}.
    \end{equation}
    On the event in \cref{eq:event} (which holds with probability at least $1-\delta$), we have
    \begin{equation}
        |\Delta E_{j, \mathrm{max}}-\bar{\Delta E}_{j, \mathrm{max}}|\leq \frac{1+\sqrt{2}}{2}\epsilon_F
    \end{equation}
    for all $1\leq j \leq K$. Therefore, guaranteeing that  $|\Delta E_{j, \mathrm{max}}-\bar{\Delta E}_{j, \mathrm{max}}|\leq\epsilon$ necessitates $\epsilon_F\leq \frac{2}{1+\sqrt{2}}\epsilon$ and therefore
    \begin{equation}
        N=\frac{(3+2\sqrt{2})8\log(4K/\delta)}{9\epsilon^2}4^m3^lM^2c^2_\mathrm{max}
    \end{equation}
    independent classical shadows. 
    Furthermore, we have
    \begin{equation}
        \left|\bar{\Delta E}_j(\theta)-\Delta E_j(\theta)\right|=\frac{1}{2}\left|\bar{A}_j-A_j-\left(\bar{A}_j-A_j\right)\cos(\theta)-\left(\bar{B}_j-B_j \right)\sin(\theta)\right|.
    \end{equation}
    Maximizing over $\theta\in[0,2\pi)$ yields
    \begin{equation}
        \max_{\theta}\left|\bar{\Delta E}_j(\theta)-\Delta E_j(\theta)\right|= \frac{\left|A_j-\bar{A}_j\right|+\sqrt{\left(A_j-\bar{A}_j\right)^2+\left(B_j-\bar{B}_j\right)^2}}{2}.
    \end{equation}
    Therefore, using $\bar{\Delta E}_{j,\mathrm{max}}=\bar{\Delta E}_j(\bar{\theta}^*_j)$, it follows
    \begin{equation}
        \left|\Delta E_j(\bar{\theta}^*_j)-\bar{\Delta E}_{j,\mathrm{max}}\right|\leq \epsilon.
    \end{equation}
    Equivalently, using $\Delta E_{j,\mathrm{max}}={\Delta E}_j({\theta}^*_j)$, it follows
    \begin{align}
        \left|\Delta E_j(\bar{\theta}^*_j)-\Delta E_{j,\mathrm{max}}\right|&=\Delta E_j(\theta^*_j)-\Delta E_j(\bar{\theta}^*_j)\\&=\Delta E_j(\theta^*_j)-\bar{\Delta E}_j({\theta}^*_j)+\bar{\Delta E}_j({\theta}^*_j)-\Delta E_j(\bar{\theta}^*_j)\\&\leq \left|\Delta E_j(\theta^*_j)-\bar{\Delta E}_j({\theta}^*_j)\right|+\left|\bar{\Delta E}_j(\bar{\theta}^*_j)-\Delta E_j(\bar{\theta}^*_j)\right|\\
        &\leq 2\epsilon,
    \end{align}
    where we have added and subtracted $\bar{\Delta E}_j(\theta^*_j)$ and used that $\bar{\Delta E}_j(\bar{\theta}^*_j)\geq\bar{\Delta E}_j({\theta}^*_j)$. Therefore, we can conclude that, with probability $1-\delta$, applying $U_j(\bar{\theta}^*_j)$ to $\ket{\psi}$ will yield an actual energy decrease that is at most $\epsilon$ smaller than the estimated maximal energy difference $\bar{\Delta E}_{j,\mathrm{max}}$ and at most $2\epsilon$ smaller than the true maximal energy difference $\Delta E_{j,\mathrm{max}}$.
\end{proof}
Notably, once $A_j$ and $B_j$ are estimated up to additive error $\frac{2}{1+\sqrt{2}}\epsilon$, the function $\Delta E_j(\theta)$ can be evaluated uniformly in $\theta\in[0,2\pi)$ up to additive error $\epsilon$. Since SEGQE only requires the maximal energy differences and the corresponding optimal parameters, we do not state this explicitly in \cref{theorem2}. 
\begin{proof}[Proof of \Cref{cor:pauli-rotations}]
Pauli operators $P\in\mathcal P^n$ satisfy $P^2=I$ and are Hermitian, and therefore belong to the class of generators considered in \cref{theorem2}. We specialize \cref{theorem2} to the case where $X_j=P_j$ ranges over all $m$-local Pauli operators. The key simplification in this setting is that conjugation by a Pauli operator does not mix Pauli strings.
Indeed, for any Pauli operators $P_j$ and $P_i$ that anticommute and have localities $m_j$ and $m_i$, one has $P_j P_i P_j = - P_i$ and $i[P_j,P_i] = 2P_jP_i$ (up to a phase), where $P_jP_i$ is a single Pauli operator of locality at most $m_i+m_j-1$. As a consequence, for Pauli rotations the coefficients $A_j$ and $B_j$ defined in the proof of \cref{theorem2} decompose into sums of $M_j$ Pauli expectation values instead of $M_j4^{m_j}$ Pauli terms as in the general case. Therefore, when calculating upper bounds on the variance and magnitude of the single-shot estimators $\hat{A}_j$ and $\hat{B}_j$, we can use a factor $3^m$ instead of $4^m$. Notably, the quantity $A_j$ does not depend on the locality of $P_j$. Accordingly, its variance and single-shot bounds are independent of $m$. For consistency with the bounds on $\hat B_j$, we retain the looser uniform bounds stated above. 
Finally, the total number of $n$-qubit, $m$-local Pauli generators is $4^m\binom{n}{m}$. Applying a union bound over this set and proceeding exactly as in the proof of \cref{theorem2} gives the stated sample complexity, which is asymptotically smaller by a factor of $\left(\frac{3}{4}\right)^m$ than in the general case considered in \cref{theorem2}.
\end{proof}

\section{Transverse-Field Ising model}
\label{app:tfi}
Consider the Hamiltonian of the TFI model defined as
\begin{equation}
    H=w\sum^n_{i=1}Z_i+J\sum^{n-1}_{i=1}X_iX_{i+1},
\end{equation}
where $w$ denotes the strength of the transverse field, $J$ the nearest neighbor coupling constant, and $n$ the number of spins (qubits). We show that, for this specific Hamiltonian, the worst-case per-iteration measurement cost of SEGQE (taking all two-local Pauli rotations as gate set) can be upper bounded by constants that are substantially smaller than the general worst-case bound of \cref{cor:pauli-rotations}.

Following the proof of \cref{theorem2}, for a two-local Pauli generator $P$ we define
\begin{equation}
    A_P := \bra{\psi}H\ket{\psi} - \bra{\psi}PHP\ket{\psi},
    \qquad
    B_P := \bra{\psi} i[P,H]\ket{\psi},
\end{equation}
and let $\hat A_P$ and $\hat B_P$ denote the corresponding single-shot classical-shadow estimators. We bound their variance from above by their second moment and bound the second moments of two Pauli estimators $\hat{p}_\alpha$ and $\hat{p}_\beta$ corresponding to the Pauli operators $P_\alpha$ and $P_\beta$ by 
\begin{equation}
    \mathbb{E}\!\left[\hat P_\alpha\,\hat P_\beta\right]
    \;\le\;
    \begin{cases}
        3^{L} & \text{if $P_\alpha$ and $P_\beta$ commute qubit-wise,}\\
        0     & \text{otherwise,}
    \end{cases}
\end{equation}
where $L$ denotes the number of qubits on which both $P_\alpha$ and $P_\beta$ act nontrivially.
Due to the locality of $H$, the expansions of $A_P$ and $B_P$ involve only Hamiltonian terms overlapping the support of $P$. For two-local generators, the resulting covariance structures depend only on the relative separation of the two sites, and it therefore suffices to consider nearest neighbor and next-nearest neighbor choices. As an example, assume $n\geq 5$ and take the next-nearest neighbor generator $P = Y_2Y_4$. A direct computation gives
\begin{equation}
    A_P = 2w(Z_2+Z_4)+2J(X_1X_2+X_2X_3+X_3X_4+X_4X_5).
\end{equation}
Ordering the six Pauli terms in $A_P$ as $(Z_2, Z_4, X_1X_2, X_2X_3,X_3X_4,X_4X_5)$, the above second-moment bound implies
\begin{equation}
    \mathrm{Var}[\hat A_P] \;\le\; f^\top K f,
\end{equation}
with $f^\top = 2(w,w,J,J,J,J)$ and
\begin{equation}
    K=\begin{pmatrix}
        3&1&0&0&1&1\\
        1&3&1&1&0&0\\
        0&1&9&3&1&1\\
        0&1&3&9&3&1\\
        1&0&1&3&9&3\\
        1&0&1&1&3&9
    \end{pmatrix}.
\end{equation}
Consequently,
\begin{equation}
    \mathrm{Var}[\hat A_P]
    \le 32w^2 + 32wJ + 240J^2
    \le 304\,\max(w^2,J^2),
\end{equation}
which is more than an order of magnitude smaller than the corresponding worst-case prefactor implied by \cref{cor:pauli-rotations},
$\mathrm{Var}[\hat A_P]\le 3888\,\max(w^2,J^2)$. An analogous calculation yields
\begin{equation}
    \mathrm{Var}[\hat B_P] \;\le\; 648\,\max(w^2,J^2).
\end{equation}
Repeating this procedure for all two-local Pauli generators on nearest neighbor and next-nearest neighbor pairs, we find that the prefactor $648$ is the largest value attained for $\mathrm{Var}[\hat B_P]$ under the above covariance bound. Therefore, for the TFI model the model-specific worst-case constants entering the per-iteration sample complexity are smaller by a factor of $6$ compared to the general worst-case estimate of \cref{cor:pauli-rotations}. 

Furthermore, in all numerical simulations of SEGQE applied to the TFI model (see \cref{sec:TFI}), we observe that, when starting from the product state $\ket{0}^{\otimes n}$, the algorithm exclusively selects Pauli rotations generated by $XY$ and $YX$ operators (see \cref{fig:example_circuit}). 
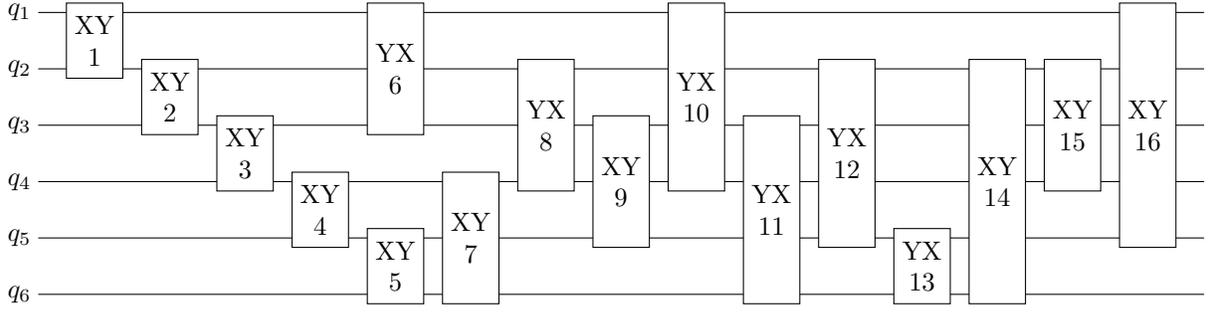
\begin{figure}
    \begin{center}
        \begin{tikzpicture} 
  \begin{pgfonlayer}{nodelayer} 
    \node [style=Label] (0) at (0, -0.0) {$q_{1}$};
    \node [style=none] (6) at (0.5, -0.0) {};
    \node [style=Label] (1) at (0, -1.5) {$q_{2}$};
    \node [style=none] (7) at (0.5, -1.5) {};
    \node [style=Label] (2) at (0, -3.0) {$q_{3}$};
    \node [style=none] (8) at (0.5, -3.0) {};
    \node [style=Label] (3) at (0, -4.5) {$q_{4}$};
    \node [style=none] (9) at (0.5, -4.5) {};
    \node [style=Label] (4) at (0, -6.0) {$q_{5}$};
    \node [style=none] (10) at (0.5, -6.0) {};
    \node [style=Label] (5) at (0, -7.5) {$q_{6}$};
    \node [style=none] (11) at (0.5, -7.5) {};
\node [style=gateTest_1] (14) at (2, -0.75) {XY\\  $1$};\node [style=gateTest_1] (17) at (4, -2.25) {XY\\  $2$};\node [style=gateTest_1] (20) at (6, -3.75) {XY\\  $3$};\node [style=gateTest_1] (23) at (8, -5.25) {XY\\  $4$};\node [style=gateTest_1] (26) at (10, -6.75) {XY\\  $5$};\node [style=gateTest_2] (29) at (10, -1.5) {YX\\  $6$};\node [style=gateTest_2] (32) at (12, -6.0) {XY\\  $7$};\node [style=gateTest_2] (35) at (14, -3.0) {YX\\  $8$};\node [style=gateTest_2] (38) at (16, -4.5) {XY\\  $9$};\node [style=gateTest_3] (41) at (18, -2.25) {YX\\  $10$};\node [style=gateTest_3] (44) at (20, -5.25) {YX\\  $11$};\node [style=gateTest_3] (47) at (22, -3.75) {YX\\  $12$};\node [style=gateTest_1] (50) at (24, -6.75) {YX\\  $13$};\node [style=gateTest_4] (53) at (26, -4.5) {XY\\  $14$};\node [style=gateTest_2] (56) at (28, -3.0) {XY\\  $15$};\node [style=gateTest_4] (59) at (30, -3.0) {XY\\  $16$};    \node [style=none] (60) at (31.5, -0.0) {};
    \node [style=none] (61) at (31.5, -1.5) {};
    \node [style=none] (62) at (31.5, -3.0) {};
    \node [style=none] (63) at (31.5, -4.5) {};
    \node [style=none] (64) at (31.5, -6.0) {};
    \node [style=none] (65) at (31.5, -7.5) {};
\end{pgfonlayer} 
 \begin{pgfonlayer}{edgelayer} 
\draw (6.center) to (60.center); 
\draw (7.center) to (61.center); 
\draw (8.center) to (62.center); 
\draw (9.center) to (63.center); 
\draw (10.center) to (64.center); 
\draw (11.center) to (65.center); 
\end{pgfonlayer} 
 \end{tikzpicture} 
    \end{center}
    \caption{Example circuit generated by SEGQE for the open-boundary transverse-field Ising model at criticality with $n=6$ qubits, starting from the initial state $\ket{0}^{\otimes n}$. Each block represents a two-qubit Pauli rotation, and its label denotes the Pauli operator generating the gate. The integer shown on each block denotes the SEGQE iteration at which the corresponding gate was appended.}
    \label{fig:example_circuit}
\end{figure}
Restricting the gate set to the Pauli rotations corresponding to these generators, allows for a further tightening of the variance bounds. In particular, we obtain
\[
\mathrm{Var}[\hat A_P]\le 136\,\max(w^2,J^2), \qquad
\mathrm{Var}[\hat B_P]\le 216\,\max(w^2,J^2)
\]
for $P=X_iY_{i+1}$, and
\[
\mathrm{Var}[\hat A_P]\le 144\,\max(w^2,J^2), \qquad
\mathrm{Var}[\hat B_P]\le 360\,\max(w^2,J^2)
\]
for $P=X_iY_{i+j}$ with $j>1$.

As a result, the upper bounds on the measurement cost relevant for the numerical experiments presented in \cref{sec:TFI} are more than an order of magnitude smaller than the conservative worst-case bound provided by \cref{cor:pauli-rotations}.

\end{document}